\documentclass[conference]{IEEEtran}
\usepackage{times, algorithm, algorithmic, graphicx, amsmath, amssymb}

\newtheorem{lemma}{Lemma}

\newcommand{\comment}[1]{}
\renewcommand{\algorithmicrequire}{\textbf{In:}}
\renewcommand{\algorithmicensure}{\textbf{Out:}}
\renewcommand\baselinestretch{0.9492}

\begin{document}

\title{Nearest Keyword Set Search in Multi-dimensional Datasets}
\author{
\IEEEauthorblockN{Vishwakarma Singh}
\IEEEauthorblockA{Department of Computer Science\\
University of California\\
Santa Barbara, USA\\
Email: vsingh014@gmail.com}
\and
\IEEEauthorblockN{Ambuj K. Singh}
\IEEEauthorblockA{Department of Computer Science\\
University of California\\
Santa Barbara, USA\\
Email: ambuj@cs.ucsb.edu}
}

\maketitle
\begin{abstract}

Keyword-based search in text-rich multi-dimensional datasets facilitates many novel applications and tools. In this paper, we consider objects that are tagged with keywords and are embedded in a vector space. For these datasets, we study queries that ask for the tightest groups of points satisfying a given set of keywords. We propose a novel method called ProMiSH (Projection and Multi Scale Hashing) that uses random projection and hash-based index structures, and achieves high scalability and speedup. We present an exact and an approximate version of the algorithm. Our empirical studies, both on real and synthetic datasets, show that ProMiSH has a speedup of more than four orders over  state-of-the-art tree-based techniques. Our scalability tests on datasets of sizes up to $10$ million and dimensions up to $100$ for queries having up to $9$ keywords show that ProMiSH scales linearly with the dataset size, the dataset dimension, the query size, and the result size.

\end{abstract}
\section{Introduction}
\label{intro}
Objects (e.g., images, chemical compounds, or documents) are often characterized by a collection of relevant features, and are commonly represented as points in a multi-dimensional attribute space. For example, images (chemical compounds) are represented using color (molecule) feature vectors. These objects also very often have descriptive text information associated with them, e.g., images are tagged with locations. In this paper, we consider multi-dimensional datasets where each data point has a set of keywords. The presence of keywords allows for the development of new tools for querying and exploring these multi-dimensional datasets.

In this paper, we study \textit{nearest keyword set search} (NKS) queries on text-rich multi-dimensional datasets. An NKS query is a set of user provided keywords. The top-$1$ result of an NKS query is a set of data points which contains all the query keywords and the points form the tightest cluster in the multi-dimensional space. Figure~\ref{fig:expQuery} illustrates an NKS query. The multi-dimensional points in the dataset are represented by dots. Each point has a unique identifier and is tagged with a set of keywords. For a query $Q$=$\{a, b, c\}$, the set of points $\{7$, $8$, $9\}$ contains all the query keywords $\{a, b, c\}$ and are nearest to each other compared to any other set of points containing these query keywords. Therefore, the set $\{7$, $8$, $9\}$ is the top-$1$ result for the query $Q$.


NKS queries are useful for many applications, e.g., photo-sharing social networks, web search engines, map services\footnote{http://maps.google.com}, GIS systems\footnote{http://www.geabios.com}~\cite{MDSD:2008}, subgraph search, and for geo-tagging of objects and regions~\cite{DZhang2010}. Consider a photo-sharing social network like Facebook where photos are tagged with people names and locations. These photos can be embedded in a high-dimensional feature space of texture, color, or shape~\cite{VsinghGeoClustering, VsinghQuerySpat}. Here an NKS query can find a group of similar photos which contains a set of people. NKS searches are also useful when labeled graphs are embedded in a high dimensional space (e.g., through Lipschitz embedding~\cite{LipschitzEmbed}) for ease of processing. In this case, a search for a subgraph that has the needed labels can be answered by an NKS search in the embedded  space~\cite{HuaHai2006}. NKS queries can also reveal  geographic patterns. GIS can characterize a region by a high-dimensional set of attributes, e.g., pressure, humidity, and soil types. Additionally, these regions can also be tagged with information such as diseases. An epidemiologist can use NKS queries to discover a pattern by finding a set of similar regions which contains all the diseases of her interest. 

\begin{figure}
  \centering
  \includegraphics[width=0.6\columnwidth]{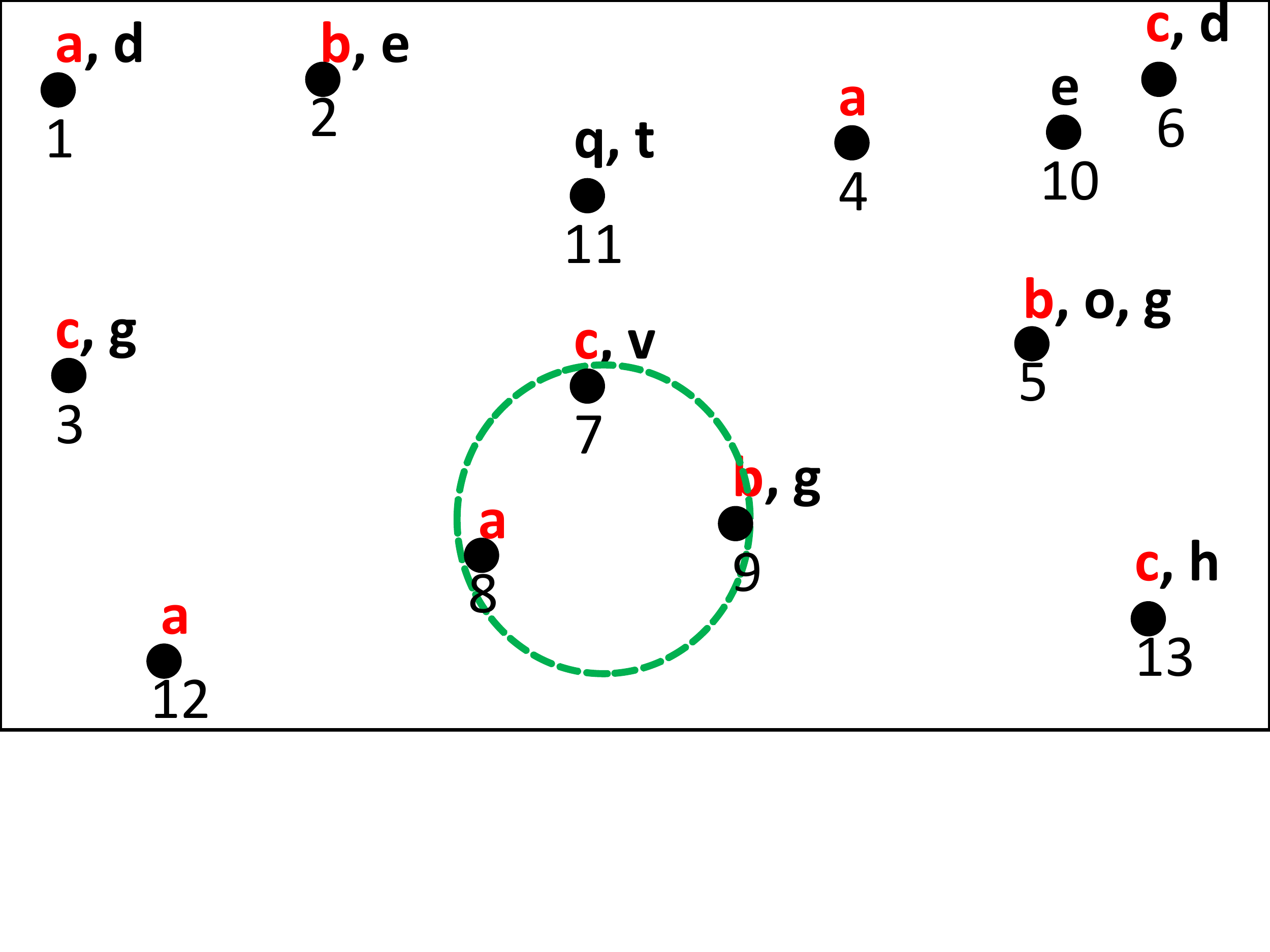}
  \vspace{-11mm}
  \caption{ An example of an NKS query on a keyword tagged multi-dimensional dataset. Query is $Q$=$\{a$, $b$, $c\}$. The top-$1$ result is the set of points $\{7$, $8$, $9\}$.}
  \vspace{-4mm}
  \label{fig:expQuery}
\end{figure}

\begin{table*}[t]
\footnotesize
\begin{center}
    \begin{tabular}{|p{4.1cm}|p{6.4cm}|p{5.9cm}|}
    \hline
    $\mathcal{D}$ : A dataset & $\mathcal{V}$ : A dictionary of unique keywords
in $\mathcal{D}$ & $Q$ : A set of keywords comprising a query\\
    \hline
    $o$ : A point in $\mathcal{D}$ & $v$ : A keyword & $N(v)$ : Number of
points in $\mathcal{D}$ having keyword $v$ \\
    \hline
    $N$ : Number of points in $\mathcal{D}$ & $U$ : Number of unique keywords in
$\mathcal{D}$ & $q$ : Number of keywords in query $Q$ \\
    \hline
    $d$ : Number of dimensions of a point & $t$ : Average number of keywords per point & $k$ :
Number of top results \\
    \hline
    $w_0$ : Initial bin-width for hashtable & $m$ : Number of unit random vectors used for projection & $L$ : Number of Hashtable-Inverted Index structures \\
    \hline
    $s$ : A scale value & $r$ : Diameter of a set of points & $z$ : A $d$-dimensional unit random vector \\
    \hline
    \end{tabular}
    \end{center}
    \vspace*{0mm}
    \caption{A glossary of notations used in the paper.}
    \vspace*{-7mm}
    \label{tab:symb}
\end{table*}

\textit{Query Definition:} Let $\mathcal{D} \subset \mathcal{R}^d$ be a $d$-dimensional dataset having $N$ points. Each point $o \in \mathcal{D}$ has a unique identifier (id). Each point is also tagged with a set of keywords $\sigma(o)$=$\{v_1, .., v_t\} \subseteq \mathcal{V}$, where $\mathcal{V}$ is a dictionary of size $U$ of all the unique keywords in $\mathcal{D}$. We use $L_2$ (Euclidean norm) to measure distance between any two points, i.e., $dist(o_i,o_j)=||o_i-o_j||_{2}$. We measure the nearness of a set of points $A$ by the maximum distance between any two points in $A$, called \textit{diameter} $r$($A$). $$r(A) = \max_{\forall {o_i,o_j \in A}} ||o_i-o_j||_{2}$$
A relatively small value of $r(A)$ implies that the corresponding objects are more similar to each other. A $q$-size NKS query $Q$=$\{v_{Q1},...,v_{Qq} \}$ has $q$ unique keywords provided by a user. Set $A \subseteq \mathcal{D} $ is a possible result, called a \textit{candidate}, of $Q$ if it contains points for all the query keywords, i.e., $ Q \subseteq \bigcup_{o \in A} \sigma(o)$, and no subset of $A$ does so. We allow overlapping candidates. If $\mathcal{S}$ is the set of all candidates of $Q$, then a result of $Q$ is the candidate $A^*$ such that $$A^* = \arg\min_{A \in \mathcal{S} }r(A).$$ A top-$k$ NKS query retrieves $k$ candidates having the least diameter.  If two candidates have equal diameters, then they are further ranked by their cardinality. 

We can also measure the nearness of a set of points $A$ by a sum of all its pairwise distances $s(A)$. Here we show with an example that $s(A)$ does not yield a tighter cluster than $r(A)$. Let $A$=$\{o_1, o_2, o_3, o_4\}$  be a set of points with following pairwise distances: $\{d(o_1, o_2)$=$2$, $d(o_1, o_3)$=$1$, $d(o_2, o_3)$=$4$, $d(o_1, o_4)$=$3$, $d(o_2, o_4)$=$3$, $d(o_3, o_4)$=$8\}$. For the set $A_1$=$\{o_1, o_2, o_3\}$ we have $r(A_1)$=$4$ and $s(A_1)$=$7$ whereas for the set $A_2$=$\{o_1, o_2, o_4\}$ we have $r(A_2)$=$3$ and $s(A_2)$=$8$. Here we see that $A_1$ has a smaller diameter whereas $A_2$ has  a smaller sum of pairwise distances. In this paper we use diameter. 


A search using a tree-based index was proposed by Zhang et al.~\cite{DZhang2010,DZhang2009} to solve NKS queries on multi-dimensional datasets. The performance of this algorithm deteriorates sharply with an increase in the dimension of the dataset as the pruning techniques become ineffective. Our empirical results show that this algorithm may take hours to terminate for a high-dimensional dataset having only a few thousands points. Authors also noted that a tree-based algorithm does not scale with the dimension of the dataset. As discussed previously, NKS queries are useful for applications of varying dimensions. Therefore, there is a need for an efficient algorithm that scales linearly with the dataset dimension and yields practical query times on large datasets.


We propose ProMiSH (\textit{Projection and Multi-Scale Hashing}) to efficiently solve NKS queries. We present an exact (\textit{ProMiSH-E}) and an approximate (\textit{ProMiSH-A}) version of the algorithm. ProMiSH-E always retrieves the true top-$k$ results, and therefore has $100\%$ accuracy. ProMiSH-A is much more time and space efficient but returns results whose diameters are within a small approximation ratio of the diameters of the true results. Both algorithms scale linearly with the dataset dimension, the dataset size, the query size, and the result size. Thus, ProMiSH possesses all the three desired characteristic of a good search algorithm: 1) high quality of results (accuracy), 2) high efficiency, and 3) good scalability.

ProMiSH-E uses a set of hashtables and inverted indices to perform a localized search of the results. ProMiSH-E hashtables are inspired from Locality Sensitive Hashing (LSH)~\cite{DatarLSH:2004}, which is a state-of-the-art method for the nearest neighbor search in high-dimensional spaces. The index structure of ProMiSH-E supports accurate search, unlike LSH-based methods that allow only approximate search with probabilistic guarantees. ProMiSH-E creates hashtables at multiple bin-widths, called scales. A search in a hashtable yields subsets of points that contain query results. ProMiSH-E explores each subset using a novel pruning based strategy. An optimal strategy is NP-Hard; therefore, ProMiSH-E  uses a greedy approach. ProMiSH-A is an approximate variation of ProMiSH-E to achieve even more space and time efficiency.


We evaluated the performance of ProMiSH on both real and synthetic datasets. We used state-of-the-art Virtual bR*-Tree \cite{DZhang2010} as a reference method for comparison. The empirical results reveal that ProMiSH consistently outperforms Virtual bR*-Tree on datasets of all dimensions. The difference in performance of ProMiSH and Virtual bR*-Tree grows to more than four orders of magnitude with an increase in the dataset dimension, the dataset size, and the query size. Our scalability tests on datasets of sizes up to $10$ million and dimensions up to $100$ for queries of sizes up to $9$ show that ProMiSH scales linearly with the dataset size, the dataset dimension, the query size, and the result size. Our datasets had as many as $24,874$ unique keywords and a data point was tagged with a maximum of $14$ keywords. The space cost analysis of the algorithms show that ProMiSH-A is much more space efficient than both ProMiSH-E and Virtual bR*-Tree.

Our main contributions are: (1) a novel multi-scale index for scalable answering of NKS queries, (2) an efficient candidate generation technique from a subset of points, and (3) extensive empirical studies. 

The paper is organized as follows. A literature survey is presented in section~\ref{subsec:survey}. Index structures are described in section~\ref{indexStruct}. An exact search algorithm (ProMiSH-E) to find subsets of points containing the results is given in section~\ref{promishE}. Section~\ref{subsec:subsetSearch} discusses how answers are generated from the subsets. The approximate algorithm (ProMiSH-A) and an analysis of its approximation ratio is presented in section~\ref{ProMiSHA}. Complexity of ProMiSH is analyzed in section~\ref{sec:costAnal}. Empirical results are presented in section~\ref{exp}. We discuss extension of ProMiSH to disk in section~\ref{disk}. Finally, we provide conclusions and future work in section~\ref{conclusions}. A glossary of the notations is shown in table~\ref{tab:symb}.

\section{Literature Survey}
\label{subsec:survey}

A variety of queries, semantically different from our NKS queries, have been studied in literature on text-rich spatial datasets. Location-specific keyword queries on the web and in the GIS systems~\cite{Zhou05, sharadSKGIS07, Vaid05spatio, ChenLi:2010} were earlier answered using a combination of R-Tree~\cite{RTree84} and inverted index. Felipe et al.~\cite{FelipeSK08} developed IR$^2$-Tree to rank objects from spatial datasets based on a combination of their distances to the query locations and the relevance of their text descriptions to the query keywords.  Cong et al.~\cite{CongSKRank09} integrated R-tree and  inverted file to answer a query similar to Felipe et al.~\cite{FelipeSK08} using a different ranking function. Martins et al.~\cite{Martins:2005} computed text relevancy and location proximity independently, and then combined the two ranking scores. Cao et al.~\cite{CaoSigmod2011} recently proposed  a method to retrieve a group of spatial web objects such that the group's keywords cover the query's keywords and the objects in the group are nearest to the query location and have the lowest inter-object distances. Other keyword-based queries on spatial datasets are aggregate nearest keyword search in spatial databases~\cite{AggNKS2010}, top-$k$ preferential query~\cite{topkPrefQueries}, finding top-$k$ sites in a spatial data based on their influence on feature points~\cite{XiaTopTInfluential:2005}, and optimal location queries~\cite{ZhangOptimalLocation:2005,ZhangOptimalLocation:2008}.



Our NKS query is similar to the $m$-closest keywords query of Zhang et al.~\cite{DZhang2009}. They designed bR*-Tree based on a R*-tree~\cite{R*Tree:1990} that
also stores bitmaps and minimum bounding rectangles (MBRs) of keywords in every node along with points MBRs. The candidates are generated by the \emph{apriori} algorithm~\cite{Apriori}. They prune unwanted candidates based on the distances between MBRs of points or keywords and the best found diameter. Their pruning techniques become ineffective with an increase in the dataset dimension as there is a large overlap between MBRs due to the curse of dimensionality. This leads to an exponential number of candidates and large query times. A poor estimation of starting diameter further worsens the performance of their algorithm. bR*-Tree also suffered from a high storage cost, therefore Zhang et al. modified bR*-Tree to create Virtual bR*-Tree~\cite{DZhang2010} in memory at run time. Virtual bR*-Tree is created from a pre-stored R*-Tree  which indexes all the points, and an inverted index which stores keyword information and path from the root node in R*-Tree for each point. Both bR*-Tree and Virtual bR*-Tree, are structurally similar, and use similar candidate generation and pruning techniques. Therefore, Virtual bR*-Tree shares similar performance weaknesses as bR*-Tree.

Tree-based indices, e.g., R-Tree~\cite{RTree84} and  M-Tree~\cite{CiacciaMTRee:1997}, have been researched extensively for an efficient near neighbor search in high-dimensional spaces. These indices fail to scale to dimensions greater than $10$ because of the curse of dimensionality~\cite{WeberVAFile:1998}. VA-file~\cite{WeberVAFile:1998} and iDistance~\cite{iDistance} provide better scalability with the dataset dimension. However, the task of designing an efficient method for solving NKS queries by adapting VA-file or iDistance is not obvious.

Random projections~\cite{Lindenstrausslemma} with hashing~\cite{Kleinberg:1997, GionisLSH:1999,DatarLSH:2004, vsingh2012, keYi2009} has come to be the state-of-the-art method for an efficient near neighbor search in high-dimensional datasets. Datar et al.~\cite{DatarLSH:2004} used random vectors constructed from $p$-stable distributions to project points, and then computed hash keys for the points by splitting the line of projected values into disjoint bins. They concatenated hash keys obtained for a point from $m$ random vectors to create a final hash key for the point. All points were indexed into a hashtable using their hash keys. Our index structure is inspired from the same.

Multi-way distance joins of a set of multi-dimensional datasets, each of which is indexed into a R-Tree, have been studied in literature~\cite{spatialJoin2,MWSJPapadias:1999}. As discussed above, a tree-based index fails to scale with the dimension of the dataset. Further, it is not straightforward to adapt these algorithms if every query requires a multi-way distance join only on a subset of the points of each dataset.

\begin{figure}
    \centering
    \includegraphics[width=0.63\columnwidth]{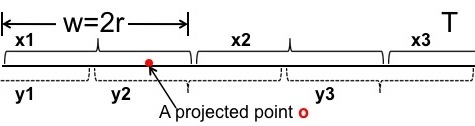}
    \vspace{-3mm}
    \caption{Division of projected values of points on a unit random vector into overlapping bins of equal width $w$=$2r$.}
    \vspace{-5mm}
    \label{divPro}
\end{figure}

\begin{figure}
  \centering
  \includegraphics[height=2.7in,width=0.98\columnwidth]{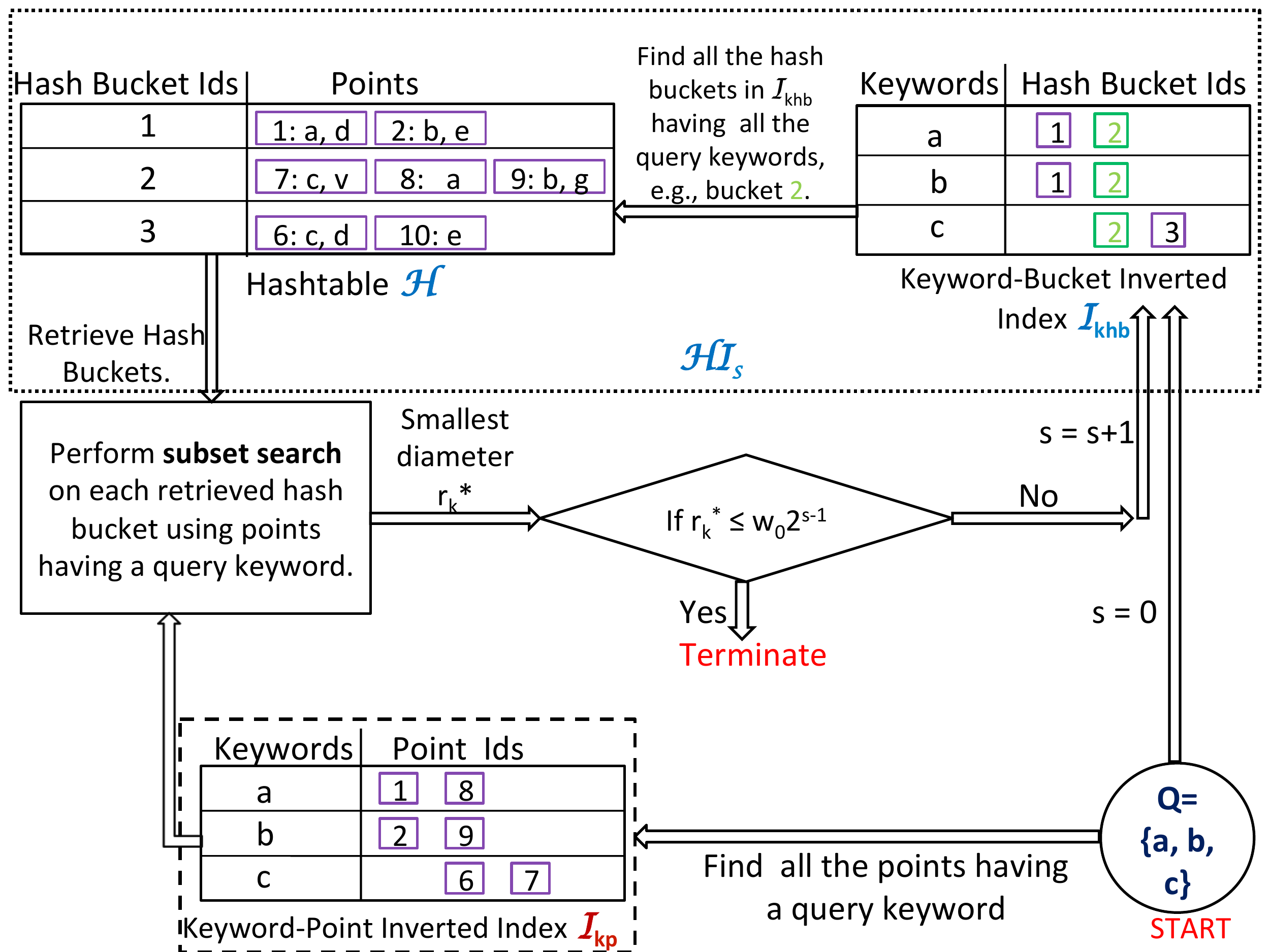}
  \vspace{-2mm}
  \caption{Index structure and flow of execution of ProMiSH.}
  \vspace{-5mm}
  \label{fig:searchAlgo}
\end{figure}

\section{Index for Exact search}
\label{indexStruct}
In this section, we describe the index structure of ProMiSH-E. It has two main data structures. The first data structure is a keyword-point inverted index $\mathcal{I}_{kp}$ that indexes all the points in the dataset $\mathcal{D}$ using their keywords. $\mathcal{I}_{kp}$ is shown with a dashed rectangle in figure~\ref{fig:searchAlgo}. The second data structure consists of multiple hashtables and their corresponding inverted indices. We call a hashtable $\mathcal{H}$ together with its corresponding inverted index $\mathcal{I}_{khb}$ as a $\mathcal{HI}$ structure.

We create a hashtable $\mathcal{H}$ as follows. We randomly choose $m$ $d$-dimensional unit vectors. We compute  projection $z.o$ of each point $o$ in $\mathcal{D}$ on each unit random vector $z$. Next, we split each line of projected values into consecutive overlapping bins of width $w$ as shown in figure~\ref{divPro}. Here a bin is equally overlapped by two other bins. We assign each point $o$ a hash key based on the bin in which it lies. Since the line is split into overlapping bins, each point $o$ lies in two bins, and therefore gets two hash keys $\{b_{1}, b_{2}\}$ from each unit random  vector $z$.  For example, the line of projected values T in figure~\ref{divPro} has been split into overlapping bins $\{$x$1$, x$2$, x$3$, y$1$, y$2$, y$3\}$. Point $o$ lies in bins x$1$ and y$2$, and therefore gets two hash keys corresponding to each of the bins. We compute hash keys using equations~\ref{hashFn1} and \ref{hashFn2}: 
\begin{eqnarray}
    \label{hashFn1}
    \mathbf{h}_1(o)  &=& \lfloor \frac {z.o} {w} \rfloor \\
    \label{hashFn2}
    \mathbf{h}_2(o)  &=& \lfloor \frac {z.o - \frac {w} {2}} {w} \rfloor + C
\end{eqnarray}
where C is a constant to distinguish values of $\mathbf{h}_1$ and $\mathbf{h}_2$. A value of C can be $(max(\mathbf{h}_1)-min(\mathbf{h}_1)+2)$.

We get $m$ pairs of hash keys for each data point $o$ using $m$ unit random vectors. We take a cartesian product of these $m$ pairs of hash keys to generate $2^m$ signatures for each point $o$. A signature $sig(o)$=$\{b_{j1}, ..., b_{jm}\}$ of a point $o$ contains a hash key from each of the $m$ pairs. For example, let $z_1$ and $z_2$ be two unit random vectors for $m$=$2$. Let the hash keys of a point $o$ be $\{$x$_1$, y$_1\}$ from $z_1$ and $\{$x$_2$, y$_2\}$ from $z_2$. ProMiSH creates $2^2$ $2$-size signatures $\{$x$_1$x$_2$, x$_1$y$_2$, y$_1$x$_2$, y$_1$y$_2\}$ for $o$ by a cartesian product. 

We hash each point $o$ using each of its $2^m$ signatures as hash key into the hashtable $\mathcal{H}$. A signature $sig(o)$ of a point $o$ is converted into a hashtable bucket identifier (bucket id) using a standard hash function, e.g., $|(\sum b_{ji}*pr_i)| \% hashtable\_size$, where $pr_i$ is a random prime number. We store a point just by its id in the hash bucket.

For each hashtable $\mathcal{H}$, we create a corresponding inverted index $\mathcal{I}_{khb}$. For each bucket of $\mathcal{H}$, we compute the union of keywords of its points. Then, we index each bucket of the hashtable $\mathcal{H}$ against each of the unique keywords it contains in the inverted index $\mathcal{I}_{khb}$.

We show a $\mathcal{HI}$ structure in figure~\ref{fig:searchAlgo} with a dotted rectangle. We create $\mathcal{HI}_s$ structures for increasing bin-width $w$=$w_02^s$, where $w_0$ is initial bin-width and $s \in \{0,...,L-1\}$ is the scale. If $pMax$ is the maximum span of projected values of points on any unit random vector, then
\begin{equation}\label{eqtn:L}
L=\lceil log_2 \left( \frac {pMax}{w_0} \right) \rceil.
\end{equation}
 
\section{Exact Search (ProMiSH-E)}
\label{promishE}
Here we describe the algorithm ProMiSH-E to find subsets of points that contain the true query results. First, we introduce lemmas which guarantee that ProMiSH-E always retrieves the true top-$k$ results using the index structure. Then, we describe the steps of ProMiSH-E to find the subsets. The algorithm to find results from these subsets is described in section~\ref{subsec:subsetSearch}.
 
\begin{lemma} \label{lemmaL2} Let $\mathcal{R}^d$ be a $d$-dimensional Euclidean space.
 Let $z$ be a vector uniformly picked
from a unit $(d$-$1)$-sphere such that $z \in \mathcal{R}^d$ and $||z||_{2}=1$. For any two points $o_1$ and
$o_2$  in  $\mathcal{R}^d$,  we have $||o_1 - o_2||_{2} \geq ||z.o_1 - z.o_2||_{2}.$
\end{lemma}
\begin{proof}
Since, an Euclidean space with dot product is an inner product space, we have
\begin{eqnarray*}
  ||z.o_1 - z.o_2||_{2} &=& |z.(o_1 - o_2)| \\
     & \le & ||z||_{2} \times ||o_1 - o_2||_{2} \\
     &=& ||o_1 - o_2||_{2} \text{ since } \|z\|_{2}=1
\end{eqnarray*}
The inequality follows from Cauchy-Schwarz inequality.
\end{proof}

\begin{lemma} \label{lem:optimality}
If a set of points $A=\{o_1, ..., o_n\}$ in $\mathcal{R}^d$ with diameter $r$ is projected onto a $d$-dimensional unit random vector $z$, and the line is split into overlapping bins of equal width $w \ge 2r$, then there exists a bin containing all the points of set $A$.
\end{lemma}
\begin{proof}
From lemma~\ref{lemmaL2} and the definition of diameter, we have $\forall o_i,\;o_j \in A,\; |z.o_i-z. o_j| \le ||o_i-o_j|| \le r$. Therefore, the span of projected values of
the points in set $A$, i.e., $\max(z.o_1, ..., z.o_n) - \min(z.o_1, ..., z.o_n)$, is  $\le
r$. Since the line is split into overlapping bins of width $2r$, it follows from the
construction, as shown in figure~\ref{divPro}, that a line segment of width $r$ is fully
contained in one of the bins. Hence, all the points in set $A$ will lie in the same bin.
\end{proof}

We illustrate here with an example how lemma~\ref{lem:optimality} guarantees retrieval of the true results. For a query $Q$, let the diameter of its top-$1$ result be $r$. We project all the data points in $\mathcal{D}$ on a unit random vector and split the projected values into overlapping bins of bin-width $2r$. Now, if we perform a search in each of the bins independently, then lemma~\ref{lem:optimality} guarantees that the top-$1$ result of query $Q$ is found in one of the bins.

A flow of execution of ProMiSH-E is shown in figure~\ref{fig:searchAlgo}. A search starts with the $\mathcal{HI}$ structure at scale $s$=$0$. ProMiSH-E finds buckets of hashtable $\mathcal{H}$, each of which contains all the query keywords, using the inverted index $\mathcal{I}_{khb}$. Then, ProMiSH-E explores each selected bucket using an efficient pruning based technique to generate results. ProMiSH-E terminates after exploring $\mathcal{HI}$ structure at the smallest scale $s$ such that the $k$th result has the diameter $r_k^* \le w_02^{s-1}$.

\begin{algorithm}[h!]
\footnotesize
\caption{ProMiSH-E}
\label{alg:ProMiSHE}
\begin{algorithmic}[1]
\REQUIRE $Q$: query keywords; $k$: number of top results
\REQUIRE $w_0$: initial bin-width
\STATE   $PQ \leftarrow [e([\;],+\infty)]$: priority queue of top-$k$ results
\STATE   $HC$: hashtable to check duplicate candidates
\STATE   $BS$ : bitset to track points having a query keyword
\FORALL  {$o \in \cup_{\forall v_Q \in Q} \mathcal{I}_{kp}[v_Q]$}
\STATE   $BS[o] \leftarrow$ true /* Find points having query keywords*/
\ENDFOR
\FORALL  {$s \in \{0, ..., L-1\}$}
\STATE   Get $\mathcal{HI}$ at $s$
\STATE   $E[\;] \leftarrow 0$ /* List of hash buckets */
\FORALL  {$v_Q \in Q$}
\FORALL  {$bId \in \mathcal{I}_{khb}[v_Q]$}
\STATE   $E[bId] \leftarrow E[bId]+1$
\ENDFOR
\ENDFOR
\FORALL {$i \in (0, ..., SizeOf(E)) $ }
\IF     {$E[i]= SizeOf(Q)$}
\STATE  $F' \leftarrow \emptyset$ /* Obtain a subset of points */
\FORALL {$o \in \mathcal{H}[i]$}
\IF     {$BS[o]=$ true}
\STATE  $F' \leftarrow F' \cup o$
\ENDIF
\ENDFOR
\IF   {checkDuplicateCand($F',\; HC$) = false}
\STATE searchInSubset($F'$, $PQ$)
\ENDIF
\ENDIF
\ENDFOR
\STATE /* Check termination condition */
\IF {$PQ[k].r \le w_02^{s-1}$}
\STATE Return $PQ$
\ENDIF
\ENDFOR
\STATE /* Perform search on $\mathcal{D}$ if algorithm has not terminated */
\FORALL {$o \in \mathcal{D}$}
\IF     {$BS[o]=$ true}
\STATE  $F' \leftarrow F' \cup o$
\ENDIF
\ENDFOR
\STATE searchInSubset($F'$, $PQ$)
\STATE Return $PQ$
\end{algorithmic}
\end{algorithm}

Algorithm~\ref{alg:ProMiSHE} details the steps of ProMiSH-E. It maintains a bitset $BS$. For each $v_Q \in Q $, ProMiSH-E retrieves the list of points corresponding to $v_Q$ from $\mathcal{I}_{kp}$ in step $4$. For each point $o$ in the retrieved list, ProMiSH-E marks the bit corresponding to $o$'s identifier in $BS$ as true in step $5$. Thus, ProMiSH-E finds all the points in $\mathcal{D}$ which are tagged with at least one query keyword. Next, the search continues in the $\mathcal{HI}$ structures, beginning at $s$=$0$. For any given scale $s$, ProMiSH-E accesses the $\mathcal{HI}$ structure created at the scale in step $8$. ProMiSH-E retrieves all the lists of hash bucket ids corresponding to keywords in $Q$ from the inverted index $\mathcal{I}_{khb}$ in steps ($10$-$11$). An intersection of these lists yields a set of hash buckets each of which contains all the query keywords in steps ($12$-$16$). For the example in figure~\ref{fig:searchAlgo}, this intersection yields the bucket id $2$. For each selected hash bucket, ProMiSH-E retrieves all the points in the bucket from the  hashtable $\mathcal{H}$. ProMiSH-E filters these points using bitset $BS$ to get a subset of points $F'$ in steps ($17$-$22$). Subset $F'$ contains only those points which are tagged with at least one query keyword and is explored further. 

Subset $F'$  is checked whether it has been explored earlier or not using \textit{checkDuplicateCand} (Algorithm \ref{alg:checkCand}) in step $23$. Since each point is hashed using $2^m$ signatures, duplicate subsets may be generated. If $F'$ has not been explored earlier, then ProMiSH-E performs a search on it using \textit{searchInSubset} (Algorithm \ref{alg:searchInSubset}) in step $24$. Results are inserted into a priority queue $PQ$ of size $k$. Each entry $e([\;],r)$ of $PQ$ is a tuple containing a set of points and the set's diameter. $PQ$ is initialized with $k$ entries, each of whose set is empty and the diameter is $+\infty$. Entries of $PQ$ are ordered by their diameters. Entries with equal diameters are further ordered by their set sizes. A new result is inserted into $PQ$ only if its diameter is smaller than the $k$th smallest diameter in $PQ$. If ProMiSH-E does not terminate after exploring the $\mathcal{HI}$ structure at the scale $s$, then the search proceeds to $\mathcal{HI}$ at the scale $(s+1)$.

\begin{algorithm}[t!]
\footnotesize
\caption{checkDuplicateCand}
\label{alg:checkCand}
\begin{algorithmic}[1]
\REQUIRE $F'$: a subset; $HC$: hashtable of subsets
\STATE $F'\leftarrow \; sort(F')$
\STATE $pr1$: list of prime numbers;  $pr2$: list of prime numbers;
\FORALL {$o \in F'$}
\STATE $pr_1 \leftarrow $ randomSelect$(pr1)$;  $pr_2 \leftarrow $ randomSelect$(pr2)$
\STATE $h_1 \leftarrow h_1$ +  $(o \times pr_1)$; $h_2 \leftarrow h_2$ +  $(o \times pr_2)$
\ENDFOR
\STATE $h$ $\leftarrow$ $h_1h_2$;
\IF {isEmpty$(HC[h])$=false}
\IF {elementWiseMatch($F',\; HC[h]$) = true}
\STATE Return true;
\ENDIF
\ENDIF
\STATE $HC[h]$.add$(F')$;
\STATE Return false;
\end{algorithmic}
\end{algorithm}

ProMiSH-E terminates when the $k$th smallest diameter $r_k$ in $PQ$ becomes less than or equal to half of the current bin-width $w$=$w_02^{s}$ in steps ($29$-$31$). Since $r_k \le \frac{w_02^{s}}{2}$, lemma~\ref{lem:optimality} guarantees that each true candidate is fully contained in one of the bins of the hashtable, and therefore  must have been explored. If ProMiSH-E fails to terminate after exploring $\mathcal{HI}$ at all the scale levels $s \in  \{0, ...,L-1\}$, then it performs a search in the complete dataset $\mathcal{D}$ in steps ($34$-$39$).

Algorithm \textit{checkDuplicateCand} (Algorithm~\ref{alg:checkCand}) uses a hashtable $HC$ to check duplicates for a subset $F'$. Points in $F'$ are sorted by their identifiers. Two separate standard hash functions are applied to the identifiers of the points in the sorted order to generate two hash values in steps ($2$-$6$). Both the hash values are concatenated to get a hash key $h$ for the subset  $F'$ in step $7$. The use of multiple hash functions helps to reduce hash collisions. If $HC$ already has a list of subsets at $h$, then an element-wise match of $F'$ is performed with each subset in the list in steps ($8$-$9$). Otherwise, $F'$ is stored in $HC$ using key $h$ in step $13$.

\section{Search in a Subset of Data Points}
\label{subsec:subsetSearch}
We present an algorithm for finding top-$k$ tightest clusters in a subset of points. A subset is obtained from a hashtable bucket as explained in section~\ref{promishE}. Points in the subset are grouped based on the query keywords. Then, all the promising candidates are explored by a multi-way distance join of these groups. The join uses $r_k$, the diameter of the $k$th result obtained so far by ProMiSH-E, as the distance threshold. 

We explain a multi-way distance join with an example. A multi-way distance join of $q$ groups $\{g_1, ..., g_q\}$ finds all the tuples $\{o_{1,i}, ..., o_{x,j},o_{y,k}, ...,o_{q,l}\}$ such that $\forall x,y$: $o_{x,j}\in g_x$, $o_{y,k}\in g_y,$ and $||o_{x,j}-o_{y,k}||_2 \le r_k$. Figure~\ref{fig:keywordOrderSelection}(a) shows groups $\{a$, $b$, $c\}$ of points obtained for a query $Q$=$\{a$, $b$, $c\}$ from a subset $F'$. We show an edge between a pair of points of two groups if the distance between the points is at most $r_k$, e.g, an edge between point $o_1$ in group $a$ and point $o_3$ in group $b$. A  multi-way distance join of these groups finds tuples $\{o_{1}$, $o_{3}$, $o_{9}\}$ and $\{o_{10}$, $o_{3}$, $o_{9}\}$. Each tuple obtained by a multi-way join is a promising candidate for a query.

\subsection{Group Ordering}
A suitable ordering of the groups leads to an efficient candidate exploration by a multi-way distance join. We first perform a pairwise inner joins of the groups with distance threshold $r_k$. In inner join, a pair of points from two groups are joined only if the distance between them is at most $r_k$. Figure~\ref{fig:keywordOrderSelection}(a) shows such a pairwise inner joins of the groups $\{a$, $b$, $c\}$. We see from figure~\ref{fig:keywordOrderSelection}(a) that a multi-way distance join in the order $\{a$, $b$, $c\}$ explores $2$ true candidates $\{ \{o_1$, $o_3$, $o_9\}$,  $\{o_{10}$, $o_3$, $o_9\}\}$ and a false candidate $\{o_1$, $o_4$, $o_6\}$. A multi-way distance join in the order $\{a$, $c$, $b\}$ explores the least number of candidates $2$. Therefore, a proper ordering of the groups leads to an effective pruning of false candidates. Optimal ordering of groups for the least number of candidates generation is NP-hard~\cite{joinNPComplete}.

\begin{figure}[t]
\centering
\begin{tabular}{cc}
\includegraphics[height=1in, width=0.41\columnwidth]{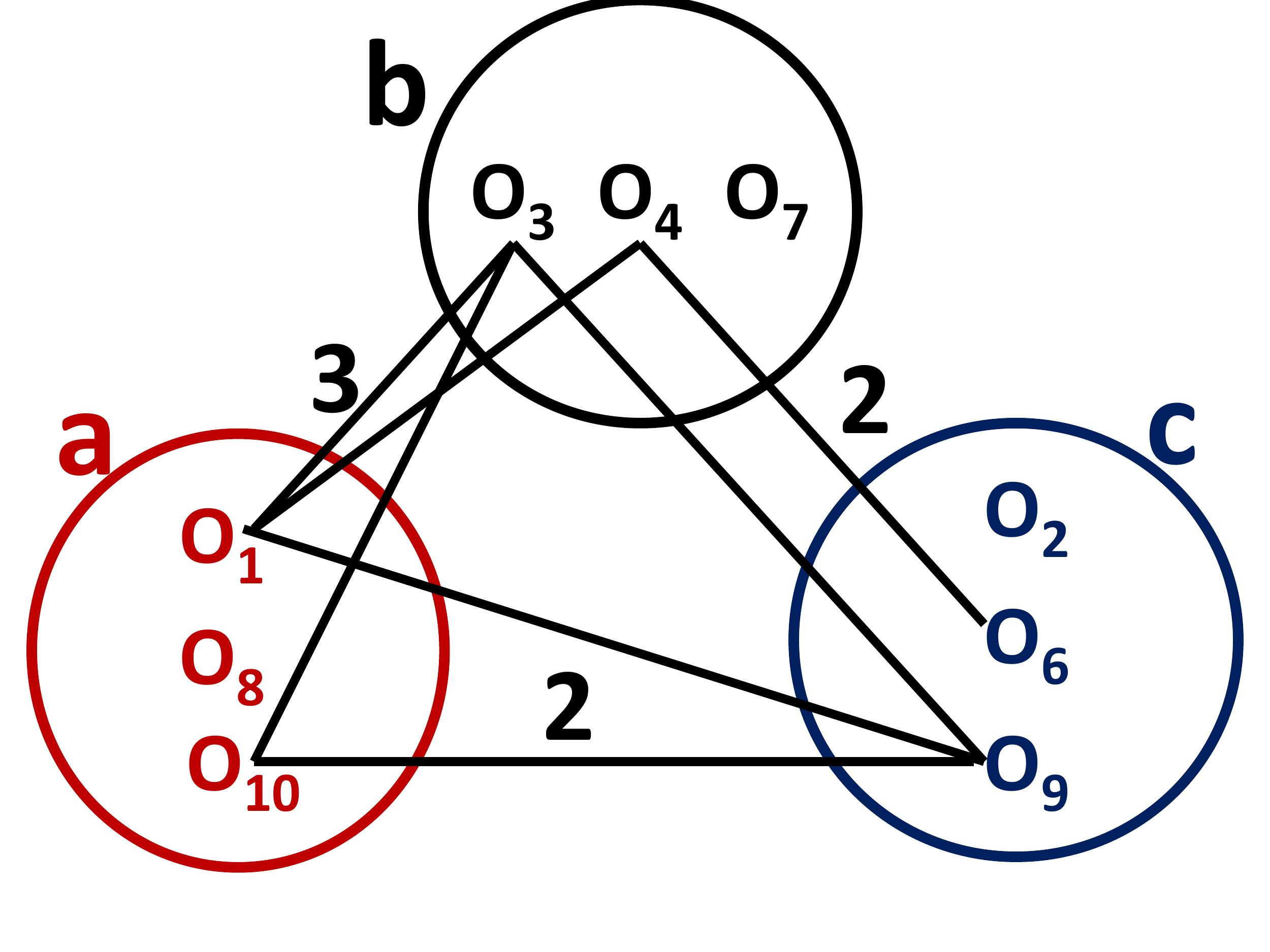} &
\includegraphics[height=0.9in, width=0.32\columnwidth]{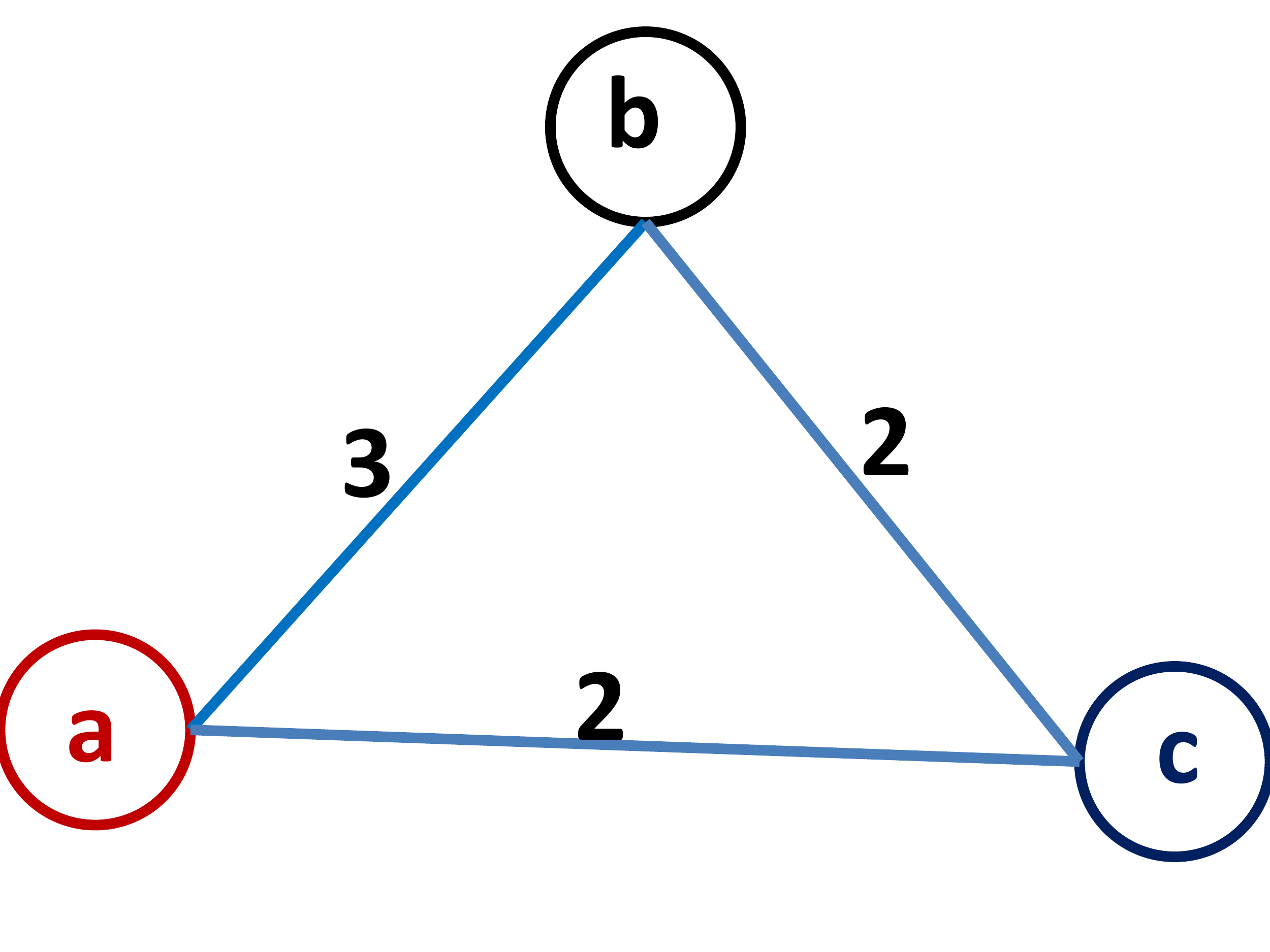}\\
(a) Pairwise inner joins  & (b) A graph representation
\end{tabular}
\vspace{-2mm}
\caption{(a) $a$, $b$, and $c$ are groups of points of a subset $F'$ obtained for a query $Q$=$\{a, b, c\}$. A point $o$ in a group $g$  is joined to a point $o'$ in another group $g'$ if $||o-o'|| \le r_k$. The groups in the order \{$a$, $c$, $b$\} generates the least number of candidates by a multi-way join. (b) A graph of pairwise inner joins. Each group is a node in the graph. The weight of an edge is the number of point pairs obtained by an inner join of the corresponding groups.}
\vspace{-5mm}
\label{fig:keywordOrderSelection}
\end{figure}

We propose a greedy approach to find the ordering of groups. We explain the algorithm with a graph in figure~\ref{fig:keywordOrderSelection}(b). Groups $\{a$, $b$, $c\}$ are nodes in the graph. The weight of an edge is the count of point pairs obtained by an inner join of the corresponding groups. The greedy method starts by selecting an edge having the least weight. If there are multiple edges with the same weight, then an edge is selected at random. Let the edge $ac$, with weight $2$, be selected in figure~\ref{fig:keywordOrderSelection}(b). This forms the ordered set $(a-c)$. The next edge to be selected is the least weight edge such that at least one of its nodes is not included in the ordered set. Edge $cb$, with weight $2$, is picked next in figure~\ref{fig:keywordOrderSelection}(b). Now the ordered set is $(a-c-b)$. This process  terminates when all the nodes are included in the set.  $(a-c-b)$ gives the ordering of the groups.

Algorithm~\ref{alg:searchInSubset} shows how the groups are ordered. The $k$th smallest diameter $r_k$ is retrieved form the priority queue $PQ$ in step $1$. For a given subset $F'$ and a query $Q$, all the points are  grouped using query keywords in steps ($2$-$5$). A pairwise inner join of the groups is performed in steps ($6$-$18$). An adjacency list $AL$ stores the distance between points which satisfy the distance threshold $r_k$. An adjacency list $M$ stores the count of point pairs obtained for each pair of groups by the inner join. A greedy algorithm finds the order of the groups in steps ($19$-$30$). It repeatedly removes an edge with the smallest weight from $M$ till all the groups are included in the order set $curOrder$. Finally, groups are sorted using $curOrder$ in step $30$.

\begin{algorithm}[!t]
\footnotesize
\caption{searchInSubset}
\label{alg:searchInSubset}
\begin{algorithmic}[1]
\REQUIRE $F'$: subset of points; $Q$: query keywords; $q$: query size
\REQUIRE $PQ$: priority queue of top-$k$ results
\STATE  $r_k \leftarrow PQ[k].r$  /* $k$th smallest diameter */
\STATE   $SL \leftarrow [(v,[\;])]$: list of lists to store groups per query keyword
\FORALL {$v \in Q $}
\STATE  $SL[v] \leftarrow \{\forall o \in F': o$ is tagged with $v\}$  /* form groups */
\ENDFOR
\STATE /* Pairwise inner joins of the groups*/
\STATE  $AL$: adjacency list to store distances between points
\STATE  $M \leftarrow 0$: adjacency list to store count of pairs between groups
\FORALL {$ (v_i, v_j) \in Q$ such that $ i \le q,\; j \le q,\; i < j $}
\FORALL {$o \in SL[v_i]$}
\FORALL {$o' \in SL[v_j]$}
\IF     {$||o-o'||_2 \le r_k$}
\STATE  $AL[o, o'] \leftarrow ||o-o'||_2$
\STATE  $M[v_i, v_j] \leftarrow M[v_i, v_j]+1$
\ENDIF
\ENDFOR
\ENDFOR
\ENDFOR
\STATE /* Order groups by a greedy approach */
\STATE $curOrder \leftarrow [\;]$
\WHILE {$Q \neq \emptyset$}
\STATE  $(v_i, v_j) \leftarrow$ removeSmallestEdge($M$)
\IF  {$v_i \not \in curOrder$}
\STATE $curOrder$.append($v_i$);  $Q \leftarrow Q \setminus v_i$
\ENDIF
\IF  {$v_j \not \in curOrder$}
\STATE $curOrder$.append($v_j$);  $Q \leftarrow Q \setminus v_j$
\ENDIF
\ENDWHILE
\STATE  sort($SL$, $curOrder$)   /* order groups */
\STATE  \textit{findCandidates}($q$, $AL$, $PQ$, $Idx$, $SL$, $curSet$, $curSetr$, $r_k$)
\end{algorithmic}
\end{algorithm}

\subsection{Nested Loops with Pruning}
We perform a multi-way distance join of the groups by nested loops. For example, consider the set of points in figure~\ref{fig:keywordOrderSelection}. Each point $o_{a,i}$ of group $a$ is checked against each point $o_{b,j}$ of group $b$ for the distance predicate, i.e.,  $||o_{a,i}-o_{b,j}||_{2} \le r_k$. If a pair ($o_{a,i}$, $o_{b,j}$) satisfies the distance predicate, then it forms a tuple of size $2$. Next, this tuple is checked against each point of group $c$. If a point $o_{c,k}$ satisfies the distance predicate with both the points $o_{a,i}$ and $o_{b, j}$, then a tuple ($o_{a,i}$, $o_{b,j}$, $o_{c, k}$) of size $3$ is generated. Each intermediate tuple generated by nested loops satisfies the property that the distance between every pair of its points is at most $r_k$. This property effectively prunes false tuples very early in the join process and helps to gain high efficiency. A candidate is found when a tuple of size $q$ is generated. If a candidate having a diameter smaller than the current value of  $r_k$ is found, then the priority queue $PQ$ and the value of $r_k$ are updated. The new value of $r_k$ is used as distance threshold for future iterations of nested loops.



We find results by nested loops as shown in Algorithm~\ref{alg:findCandidates} (\textit{findCandidates}). Nested loops are performed recursively. An intermediate tuple $curSet$ is checked against each point of group $SL$[$Idx$] in steps ($2$-$23$). First, it is determined using $AL$ whether the distance between the last point in $curSet$ and a point $o$ in $SL$[$Idx$] is at most $r_k$ in step $3$. Then, the point $o$ is checked against each point in $curSet$ for the distance predicate in steps ($5$-$15$). The diameter of $curSet$ is updated in steps ($9$-$11$). If a point $o$  satisfies the distance predicate with each point of $curSet$, then a new tuple $newCurSet$ is formed in step $17$ by appending $o$ to $curSet$. Next, a recursive call is made to \textit{findCandidates} on the next group $SL$[$Idx+1$] with $newCurSet$ and $newCurSetr$. A candidate is found if $curSet$ has a point from every group. A result is inserted into $PQ$ after checking for duplicates in steps ($26$-$33$). A duplicate check is done by a sequential match with the results in $PQ$. For a large value of $k$, a method similar to Algorithm~\ref{alg:checkCand} can be used for a duplicate check. If a new result gets inserted into $PQ$, then the value of $r_k$ is updated in step $18$.

\begin{algorithm}[!t]
\footnotesize
\caption{findCandidates}
\label{alg:findCandidates}
\begin{algorithmic}[1]
\REQUIRE $q$: query size; $SL$: list of groups
\REQUIRE $AL$: adjacency list of distances between points
\REQUIRE $PQ$: priority queue of top-$k$ results
\REQUIRE $Idx$: group index in $SL$
\REQUIRE $curSet$: an intermediate tuple
\REQUIRE $curSetr$: an intermediate tuple's diameter
\IF { $Idx \le q $}
    \FORALL  {$o \in SL[Idx]$}
       \IF  {$AL$[$curSet$[$Idx$-$1$], $o$] $\le r_k$}
           \STATE   $newCurSetr \leftarrow curSetr$
           \FORALL  {$o' \in curSet $}
            \STATE   $dist \leftarrow $ $AL$[$o$, $o'$]
            \IF      {$dist \le r_k$}
                \STATE    $flag \leftarrow$  true
                \IF       {$newCurSetr < dist$}
                \STATE     $newCurSetr \leftarrow dist$
                \ENDIF
            \ELSE
                \STATE    $flag \leftarrow $ false; break;
            \ENDIF
          \ENDFOR
        \IF {flag = true}
            \STATE  $newCurSet \leftarrow curSet$.append($o$)
            \STATE $r_k \leftarrow$ \textit{findCandidates}($q$, $AL$, $PQ$, $Idx$+1, $SL$, $newCurSet$, $newCurSetr$, $r_k$)
        \ELSE
            \STATE Continue;
        \ENDIF
      \ENDIF
    \ENDFOR
    \STATE return $r_k$
\ELSE
    \IF  {checkDuplicateAnswers($curSet$, $PQ$) = true}
       \STATE return $r_k$
    \ELSE
        \IF  {$curSetr < PQ[k].r$}
            \STATE $PQ$.\textit{Insert}([$curSet$, $curSetr$])
            \STATE return $PQ[k].r$
        \ENDIF
    \ENDIF
\ENDIF
\end{algorithmic}
\end{algorithm}

\section{Approximate Search (ProMiSH-A)}
\label{ProMiSHA}
We present ProMiSH-A that is more space and time efficient than ProMiSH-E. We also use a statistical model to show that ProMiSH-A retrieves results within a small approximation ratio of the true results with a high probability.

The index structure and the search method of ProMiSH-A are variations of ProMiSH-E, therefore we describe only the differences. The index structure of ProMiSH-A  differs from ProMiSH-E only in the way the line of projected values of points on a unit random vector is split. ProMiSH-A splits the line into non-overlapping bins of equal width, unlike ProMiSH-E which splits the line into overlapping bins. Therefore, each data point $o$ gets one hash key from a unit random vector $z$ in ProMiSH-A.  A  signature $sig(o)$ is created for each point $o$ by the concatenation of its hash keys obtained from each of the $m$ unit random vectors. Each point is hashed using its signature $sig(o)$ into a hashtable at a given scale.

The search technique of ProMiSH-A differs from ProMiSH-E in the initialization of priority queue $PQ$ and the termination condition. ProMiSH-A starts with an empty priority queue $PQ$, unlike ProMiSH-E whose priority queue is initialized with $k$ entries. ProMiSH-A checks for a termination condition after fully exploring a hashtable at a given scale. It terminates if it has $k$ entries in its priority queue $PQ$. Since each point is hashed only once into a hashtable of ProMiSH-A, it does not perform a subset duplicate check or a result duplicate check.

\textbf{Bound on approximation ratio:} Define approximation ratio $\rho \ge 1$ as the ratio of the diameter of the result reported by ProMiSH-A $r$ to the diameter of the true result $r^*$, i.e., $\rho$=$\frac{r}{r^*}$. Let $\mathcal{D}$ be a $d$-dimensional dataset and $Q$=$\{v_{Q1}, \cdots, v_{Qq}\}$ be an NKS query. Let $f_v$ be the probability mass function of the keywords $v \in \mathcal{V}$. Using $f_v$, we get the number of points tagged with a query keyword $v_Q$ as $N(v_Q)= f_v(v_Q) \times N $. Therefore, the total number of candidates for query $Q$ in $\mathcal{D}$ is
\begin{equation}\label{eqn:totalCand}
    N_n = \prod_{i=1}^{q} f_v(v_{Qi}) \times N.
\end{equation}

Let $f_r$ be the probability mass function of diameters of candidates of $Q$. Then, the total number of candidates of query $Q$ having diameter $r$ is given by
\begin{equation}\label{eqn:totalCandr}
    N_r= f_r(r) \times N_n.
\end{equation}

We project all the points in dataset $\mathcal{D}$, which contain at least one query keyword $v_Q$, onto a unit random vector $z$. We split the line of projected values into non-overlapping bins of equal width $w$. Let $Pr(A|r)$ be the conditional probability for random unit vectors that a candidate $A$ of query $Q$ having diameter $r$ is fully contained within a bin. For $m$ independent unit random vectors, the joint probability that a candidate $A$ is contained in a bin in each of the $m$ vectors is $Pr(A|r)^m$. The probability that no candidate of diameter $r$ is retrieved by ProMiSH-A from the hashtable, created using $m$ unit random vectors, is $(1-Pr(A|r)^m)^{N_r}$. Let the diameter of the top-$1$ result of query $Q$ be $r^*$. Then, the probability $P(r')$ of at least one candidate of any diameter $r$, where $r^* \le r \le r'$, being retrieved by ProMiSH-A is given by
\begin{equation}\label{eqn:approx}
    P(r') = 1-\prod_{r=r^*}^{r'} (1-Pr(A|r)^m)^{N_r}.
\end{equation}

For a given constant $\lambda$, $0 \le \lambda \le 1$, we can compute the smallest value of $r'$ using  equation~\ref{eqn:approx} such that $\lambda \le P(r')$. The value $\rho^*$=$\frac{r'}{r^*}$ gives an upper bound on the approximation ratio of the results returned by ProMiSH-A with the probability $\lambda$. 

We empirically computed $\rho^*$ for queries of size $q$=$3$ for different values of $\lambda$ using this model. We used a $32$-dimensional real dataset having $1$ million points described in section~\ref{exp} for our study. For a set of randomly chosen queries of size $3$, we computed the values of $N_r$ and $Pr(A|r)^2$. We used projections on $1$ million random vectors and a bin-width of $w$=$100$ for computing $Pr(A|r)^2$. We obtained the approximation ratio bound of $\rho^*$=$1.4$ and $\rho^*$=$1.5$ for $\lambda$=$0.8$ and $\lambda$=$0.95$ respectively. 



\section{Complexity Analysis of ProMiSH}
\label{sec:costAnal}

\begin{figure}[t]
\centering
\begin{tabular}{cc}
\includegraphics[height=0.85in,width=0.16\textwidth]{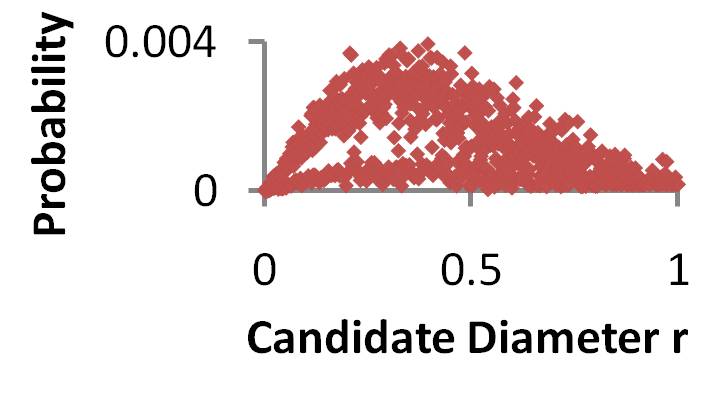} &
\includegraphics[height=0.85in,width=0.16\textwidth]{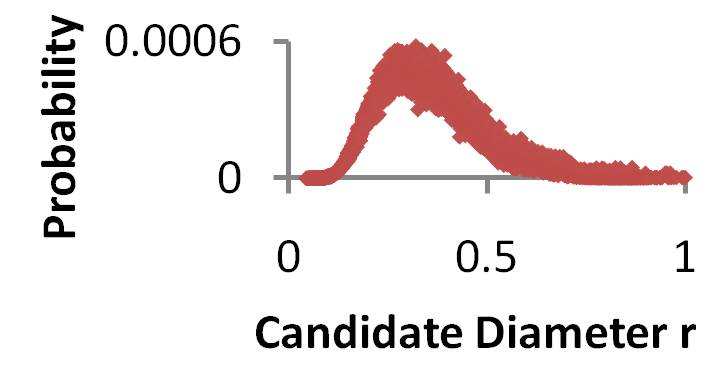}\\
(a) $d$=$2$ & (b) $d$=$16$ \\
\end{tabular}
\vspace{-2mm}
\caption{Probability mass functions $f_r$ of diameters of candidates of a query of size $3$ on a $2$-dimensional and a $16$-dimensional real datasets.}
\vspace{-1mm}
\label{fig:diaProb}
\end{figure}

\begin{figure}[t]
\centering
\begin{tabular}{cc}
\includegraphics[height=0.85in,width=0.16\textwidth]{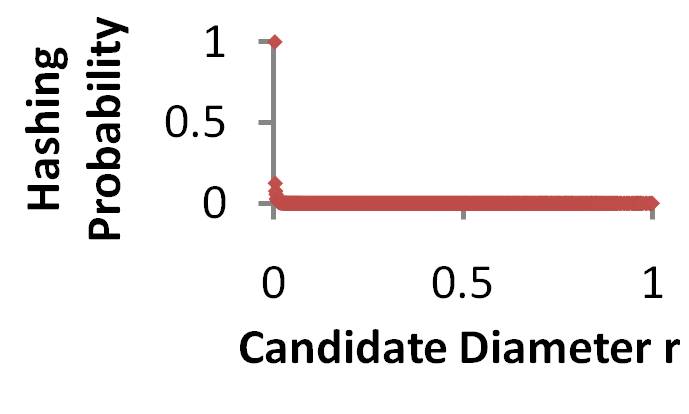} &
\includegraphics[height=0.85in,width=0.16\textwidth]{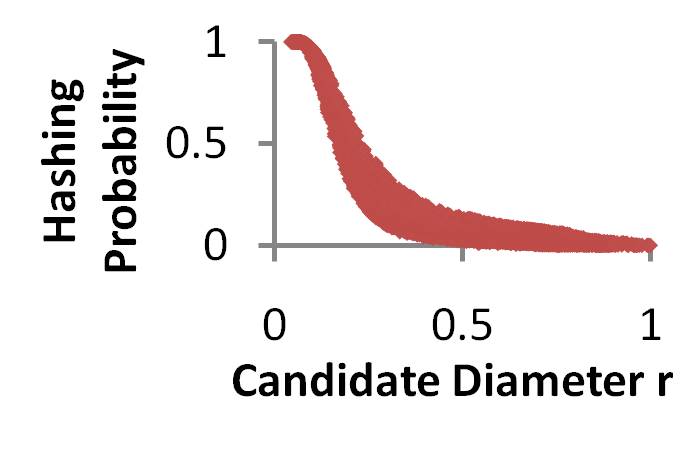}\\
(a) $d$=$2$ & (b) $d$=$16$ \\
\end{tabular}
\vspace{-2mm}
\caption{Values of $Pr(A|r)^2$ for varying diameters of candidates of a query of size $3$ on a $2$-dimensional and a $16$-dimensional real datasets.}
\vspace{0mm}
\label{fig:hashProb}
\end{figure}

\begin{table}[t]
\begin{center}
    \footnotesize
    \begin{tabular}{|c|c|c|c|c|c|}
     \hline
     Dataset Dimension $d$ & $2$ & $4$ & $8$ & $16$ & $32$ \\
     \hline
     Percentage Ratio ($\frac {N_p} {N_n}$) & 0.007 & 0.3 & 5.8 & 22 & 47 \\
     \hline
    \end{tabular}
    \end{center}
    \vspace*{0mm}
    \caption{Percentage ratio of the expected number of candidates $N_p$ to the total number of candidates $N_n$ of a query.}
    \vspace*{-8mm}
    \label{tab:canRatio}
\end{table}

We first show using a statistical model that ProMiSH effectively prunes the false candidates. Then, we analyze the time and the space complexity of ProMiSH. Let $\mathcal{D}$ be a $d$-dimensional dataset  of size $N$ where each point $o$ is tagged with $t$ keywords. Let $U$ be the number of unique keywords in $\mathcal{D}$. Let $Q$=$\{v_{Q1}, \cdots, v_{Qq}\}$ be an NKS query of size $q$.

\textbf{Statistical Model:} Let the set $A^* \subset \mathcal{D}$ with diameter $r^*$ be the top-$1$ result of query $Q$. We use  $t$=1 for our model. Let $f_v$ be the probability mass function of the keywords $v \in \mathcal{V}$. Let $f_r$ be the probability mass function of diameters of candidates of $Q$. The total number of candidates $N_n$ and $N_r$ of query $Q$ are given by equations~\ref{eqn:totalCand} and \ref{eqn:totalCandr} respectively. We select all the points in $\mathcal{D}$ which contain at least one query keyword $v_Q$. We project these points on a unit random vector $z$. We split the line of projected values into overlapping bins of equal width $w=2r^*$. Let $Pr(A|r)$ be the conditional probability for random unit vectors that a candidate $A$ of query $Q$ having diameter $r$ is fully contained within a bin. For $m$ independent unit random vectors, the joint probability that a candidate $A$ is contained in a bin in each of the $m$ vectors is $Pr(A|r)^m$. The expected number of candidates explored by ProMiSH in a hashtable, created using $m$ unit random vectors, is
\begin{equation}\label{eqn:expectedCand}
    N_p = \sum_r Pr(A|r)^m \times N_r.
\end{equation}

We empirically computed the probability mass function $f_r$, the probability $Pr(A|r)^m$, and the ratio of $N_p$ to $N_n$. We used real datasets of size $N$=$1$ million and varying dimensions for our experiments. These datasets are described in section~\ref{exp}. We used randomly selected queries of size $q$=$3$. We show probability mass functions $f_r$ of diameters of candidates of query $Q$ on datasets of dimensions $d$=$2$ and $d$=$16$ in figure~\ref{fig:diaProb}. We computed diameters of all the candidates of query $Q$ in the dataset to obtain $f_r$ and $r^*$. The diameters of the candidates were scaled to lie between $0$ and $1$. We show values of $Pr(A|r)^2$ for varying diameters of candidates of query $Q$ on datasets of dimensions $d$=$2$ and $d$=$16$ in figure~\ref{fig:hashProb}. To compute $Pr(A|r)$, we randomly chose a candidate $A$ of diameter $r$. We projected all the points of $A$ on one million unit random vectors. Then, we computed the number of vectors on each of which all the points in $A$ lie in the same bin.

We make following observations from the above analysis: (a) diameters of the candidates of a query have a heavy-tailed distribution, and (b) the value of $Pr(A|r)^m$ decreases exponentially with the diameter of the candidate of a query. The first observation implies that a large number of the candidates have diameters much larger than $r^*$. The second observation implies that the candidates with diameter larger than $r^*$ have much smaller chance of falling in a bin than $A^*$, and thus being probed by ProMiSH.  Therefore, most of the false candidates, i.e., candidates with diameters larger than $r^*$, are effectively pruned out by ProMiSH using its index.

We present the percentage ratio of $N_p$ to $N_n$ in table~\ref{tab:canRatio} for datasets of varying dimensions. Each ratio is computed as an average of $50$ random queries. We observe from table~\ref{tab:canRatio} that ProMiSH prunes more than $99\%$ of the false candidates for datasets of low dimensions, e.g., $d$=$2$. For high dimensions, e.g., $d$=$32$, more than $50\%$ of the false candidates get pruned.

\textbf{Time complexity:} We assume that the data points are uniformly distributed across all the keywords. Therefore, the total number of the data points tagged with a keyword $v$ is
\begin{equation}
N(v)= N \times (\frac {t} {U}). \nonumber
\end{equation} 

Let the index structure of ProMiSH-E be comprised of $\mathcal{HI}$ structures at $L$ scale levels where the value of $L$ is obtained by equation~\ref{eqtn:L}. Let $\mathcal{H}_s$ be the hashtable at scale $s$. We assume without any loss of generality that the hashtable $\mathcal{H}_s$ is created using $m$=$1$ unit random vector. Let $pSpan$ be the span of the projected values of the data points on the unit random vector. We assume that the data points tagged with a keyword $v$ are uniformly distributed on the line of projected values. ProMiSH-E divides the the line of projected values into overlapping bins to compute the hash keys of the points using a bin-width of $w$=$w_02^s$. Therefore, the number of the data points having keyword $v$ lying in a bucket $b$ of $\mathcal{H}_s$ is
\begin{eqnarray}
  N(vb) &=& N(v)*w/pSpan \nonumber\\
        &=&  N(v)/2^{L-s}. \nonumber
\end{eqnarray}

We first compute the cost of a search in a bucket $b$ of $\mathcal{H}_s$. The cost of pairwise inner joins for query $Q$ of size $q$ for $d$-dimensional data points is $(N(vb) \times q)^2 \times d/2$. Nested loop  enumerates the candidates by looking up the pre-computed distances between the points from the adjacency list. Therefore, the worst case cost of the nested loop is $N(vb)^q$. The total cost of a search in a bucket $b$ of the hashtable $\mathcal{H}_s$ is
\begin{equation}
T(bs) = ((N(vb) \times q)^2 \times d/2)+N(vb)^q. \nonumber
\end{equation}
The total number of buckets in  $\mathcal{H}_s$ of ProMiSH-E is $2^{L-s+1}$. Therefore, the cost of a search in $\mathcal{H}_s$ is
\begin{equation}
T(\mathcal{H}_s)= 2^{L-s+1} \times T(bs). \nonumber
\end{equation}

ProMiSH-A divides the line of projected values into non-overlap\-ping bins. The total number of buckets in  $\mathcal{H}_s$ of ProMiSH-A is $2^{L-s}$. Therefore, the cost of a search in $\mathcal{H}_s$ is
\begin{equation}
T(\mathcal{H}_s)= 2^{L-s} \times T(bs). \nonumber
\end{equation}
We present the query times of ProMiSH for NKS queries on multiple real and synthetic datasets in section~\ref{exp}.

\textbf{Space complexity:} Let the space cost of a point's identifier, a dimension of a point, and a keyword be $E$ bytes individually. The index structure of ProMiSH consists of the \textit{keyword-point} inverted index $\mathcal{I}_{kp}$ and $L$ pairs of hashtable $\mathcal{H}$ and  \textit{keyword-bucket} inverted index $\mathcal{I}_{khb}$. The space cost of $\mathcal{I}_{kp}$ is $S(\mathcal{I}_{kp})$ =($N$ $\times$ $E$ $\times$ $t$) bytes. For ProMiSH-E, each point is hashed into a hashtable  $\mathcal{H}$ using $2^m$ signatures, therefore a hashtable takes $S_E(\mathcal{H})$ =($2^m$ $\times$ $N$ $\times$ $E$) bytes. For ProMiSH-A, each point is hashed using only one signature, therefore a hashtable takes  $S_A(\mathcal{H})$ =($N$ $\times$ $E$) bytes. The space cost of a $\mathcal{I}_{khb}$ inverted index is $S(\mathcal{I}_{khb})$ = ($U$ $\times$ $M$ $\times$ $log_2M/8$) bytes, where $M$ is the number of buckets in hashtable $\mathcal{H}$. The total space cost of the index of ProMiSH-E is $S(\mathcal{I}_{kp}) + S_E(\mathcal{H}) + S(\mathcal{I}_{khb})$. The total space cost of the index of ProMiSH-A is $S(\mathcal{I}_{kp}) + S_A(\mathcal{H}) + S(\mathcal{I}_{khb})$. The ratio of index size to dataset size is further analyzed in section~\ref{sec:promishESC}.

\section{Empirical Evaluation}
\label{exp}
We evaluated the performance of ProM\-iSH-E and Pr\-oMiSH-A on synthetic and real datasets. We used recently introduced Virtual bR*-Tree~\cite{DZhang2010} as a reference method for comparison (see section~\ref{subsec:survey} for a description). We first introduce the datasets and the metrics used for measuring the performance of the algorithms. Then, we discuss the quality results of the algorithms on real datasets. Next, we describe comparative results of ProMiSH-E, ProMiSH-A, and Virtual bR*-Tree on both synthetic and real datasets. We also report scalability results of ProMiSH on both synthetic and real datasets. Finally, we present a comparison of the space usage of all the algorithms.

\begin{table}[t]
\begin{center}
    \scriptsize
    \begin{tabular}{|c|c|c|c|}
     \hline
     Id & Dataset Size ($N$) & Dictionary Size $U$ & Average $t$ \\
     \hline
     $1$ & $10$,$000$ & $5$,$661$ & $12$\\
     \hline
     $2$ & $30$,$000$ & $6$,$753$ & $13$\\
     \hline
     $3$ & $50$,$000$ & $7$,$101$ & $13$ \\
     \hline
     $4$ & $70$,$000$ & $7$,$902$ & $14$ \\
     \hline
     $5$ & $1$ Million & $24$,$874$ & $11$\\
    \hline
    \end{tabular}
    \end{center}
    \vspace*{0mm}
    \caption{Description of real datasets of five different sizes.}
    \vspace*{-5mm}
    \label{tab:realDataset}
\end{table}

\textbf{Datasets:} We used both synthetic and real datasets for experiments. Synthetic data was randomly generated. Each component of a $d$-dimensional synthetic point was chosen uniformly from [$0$-$10$,$000$]. Each synthetic point was randomly tagged with $t$ keywords. A dataset is characterized by its (1) size, $N$; (2) dimensionality, $d$; (3) dictionary size, $U$; and (4) the number of keywords associated with each point, $t$. We created various synthetic datasets by varying these parameters for our empirical studies.

Our NKS query is useful for finding tight clusters of photos which contain all the keywords provided by a user in a photo-sharing social network as discussed in section~\ref{intro}. Based on this application, we used images having descriptive tags as real datasets. We downloaded images with their textual keywords from Flickr\footnote{http://www.flickr.com/}. We transformed each image into grayscale. We created a $d$-dimensional dataset by extracting a $d$-dimensional color histogram from each image. Each data point was tagged with the keywords of its corresponding image. We describe real datasets of five different sizes used in our empirical studies in table~\ref{tab:realDataset}. The largest real dataset had $24,874$ unique keywords and each point in it was tagged with $11$ keywords. A query for a dataset was created by randomly picking a set of keywords from the dictionary of the dataset. A query is parameterized by its size $q$.


\begin{figure}
  \centering
  \includegraphics[width=0.5\columnwidth]{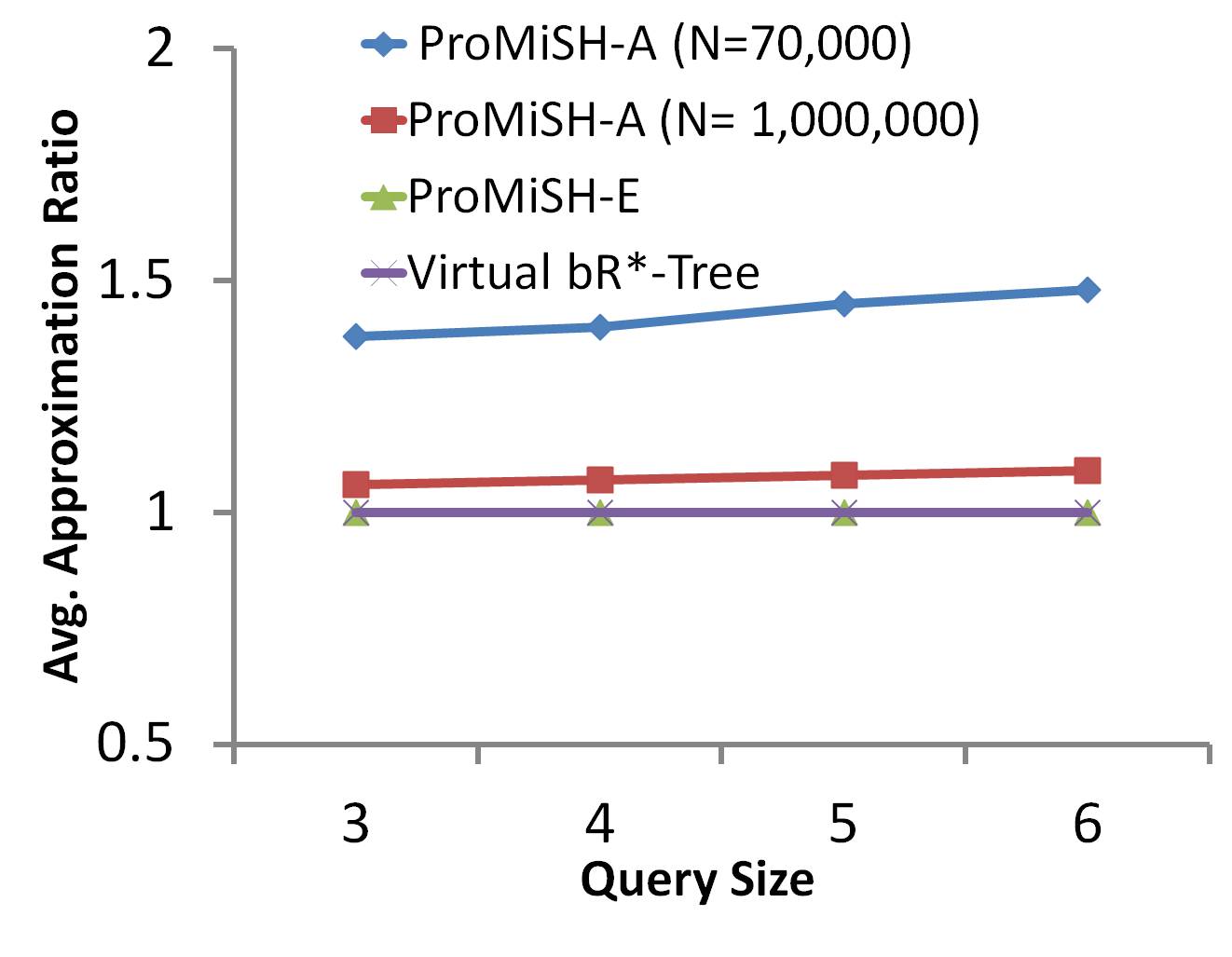}
  \vspace{-4mm}
  \caption{Average approximation ratio of ProMiSH-A for varying query sizes on $32$-dimensional real datasets of various sizes.}
  \vspace{-2mm}
  \label{fig:approx_ratio}
\end{figure}

\begin{figure}[t]
\centering
\includegraphics[width=0.45\columnwidth]{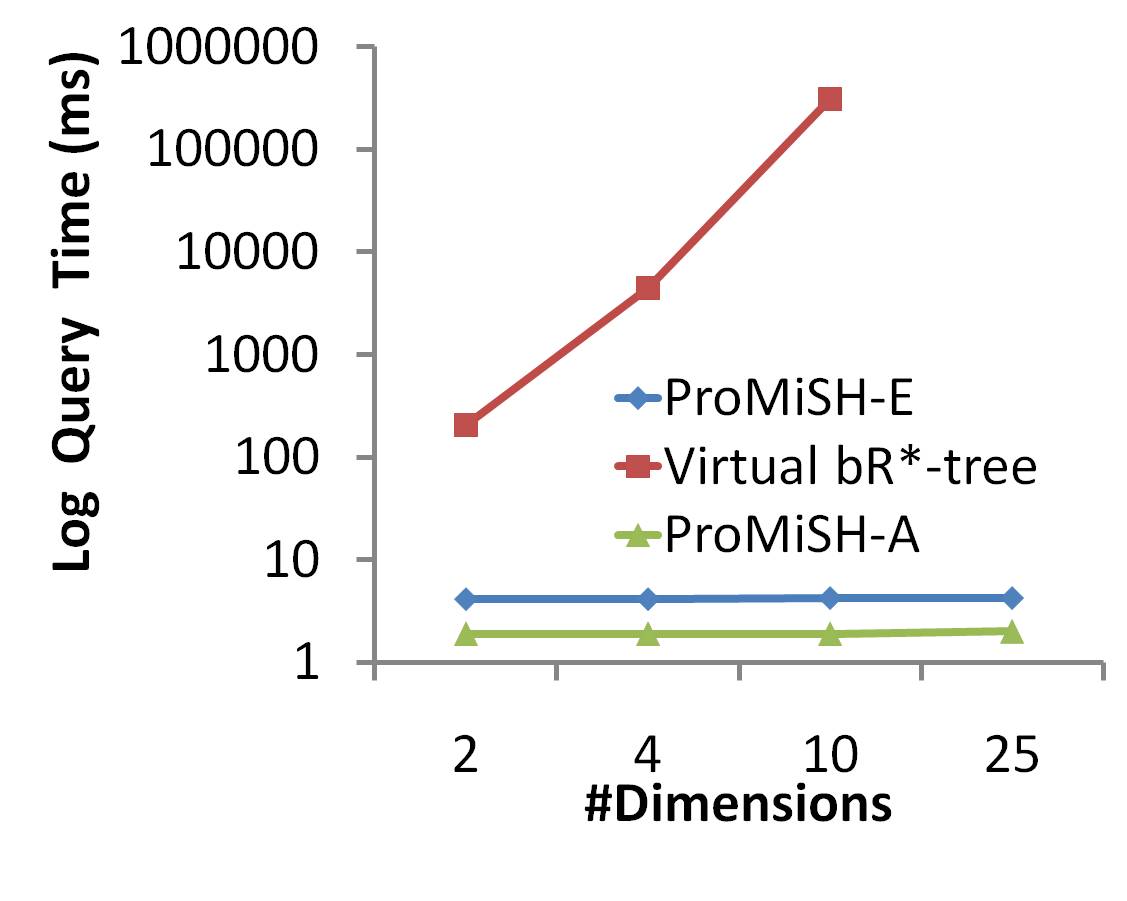}
\vspace{-4mm}
\caption{Query time comparison of algorithms for retrieving top-$1$ results for queries of size $q$=$5$ on synthetic datasets of varying dimensions $d$. Values of $N$=$100$,$000$, $t$=$1$, and $U$=$1$,$000$ were used for each dataset.}
\label{fig:syn_comp_var_d}
\vspace{-3mm}
\end{figure}



\textbf{Performance metrics:} We used \textit{approximation ratio}, \textit{query time}, and \textit{space usage} as metrics to evaluate the quality of results (accuracy), the efficiency, and the scalability of the search algorithms.

We measured the quality of results of an algorithm by its approximation ratio~\cite{GionisLSH:1999,keYi2009}. For $ 1 \le i \le k $, if $r_i$ is the $i$th diameter in top-$k$ results retrieved by an algorithm for a query $Q$ and $r_i^*$ is the true $i$th diameter, then the approximation ratio of the algorithm for top-$k$ search is given by $\rho(Q)=(\sum_{i=1}^k \frac {r_i} {r_i^*})/k$. The smaller the value of $\rho(Q)$, the better is the quality of the results returned by the algorithm. The least value of $\rho(Q)$ is $1$. We report the \textit{average approximation ratio (AAR)} for the queries of a given size, which is the mean of the approximation ratios of $50$ queries.

We validated the time efficiency of the algorithms by measuring their query times. The index structure and the dataset for each  method reside in memory. Therefore, the query time measured as the elapsed CPU time between the start and the completion of a query gives a fair comparison between the methods. A query was executed multiple times and the average execution time was taken as its query time. Finally, we report the query time for a query size $q$ as an average of $50$ different queries. The query time of a search algorithm mainly depends on the dataset size $N$, the dataset dimension $d$, and the query size $q$. Therefore, we validated the scalability of the algorithms by computing their query times for varying values of $N$, $d$, and $q$. We verified the space efficiency of an algorithm by computing the ratio of its index memory footprint to the dataset memory footprint.

\textbf{Implementation of the methods:} We implemented all the methods in Java. For Virtual bR*-Tree, we fixed the leaf node size to $1$,$000$ entries and other nodes' sizes to $100$ entries. Virtual bR*-Tree finds only the smallest subset, therefore we used $k$=$1$ for ProMiSH for a fair comparison. We used the value of $m$=$2$ and $L$=$5$ to create the index structure of ProMiSH-E and ProMiSH-A. For a dataset, if $pMax$ is the maximum span of projected values of data points on any unit random vector, then a value of $w_0$=$\frac {pMax}{2^{L}}$ was used as the initial bin-width.


All the experiments were performed on a machine having Quad-Core Intel Xeon CPU@$2$.$00$GHz, $4$,$096$ KB cache, and $98$ GB main memory and running $64$-bit Linux version $2$.$6$.

\begin{figure}[t]
\centering
\includegraphics[width=0.5\columnwidth]{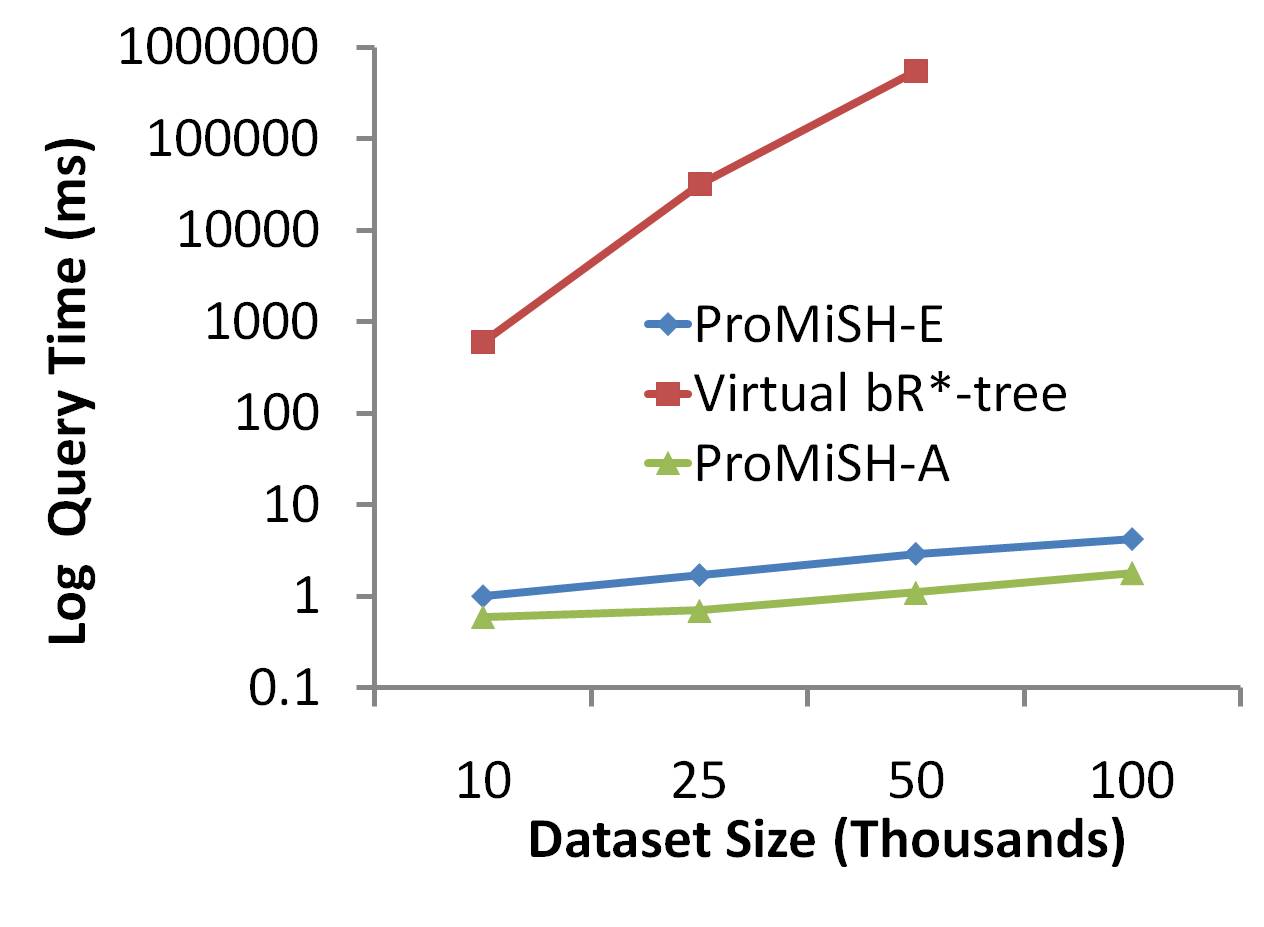}
\vspace{-4mm}
\caption{ Query time comparison of algorithms for retrieving top-$1$ results for queries of size $q$=$5$ on $25$-dimensional synthetic datasets of varying sizes $N$. Values of $t$=$1$ and $U$=$1$,$000$ were used for each dataset.}
\vspace{-2mm}
\label{fig:syn_comp_var_N}
\end{figure}

\begin{figure}[t]
\centering
\includegraphics[width=0.5\columnwidth]{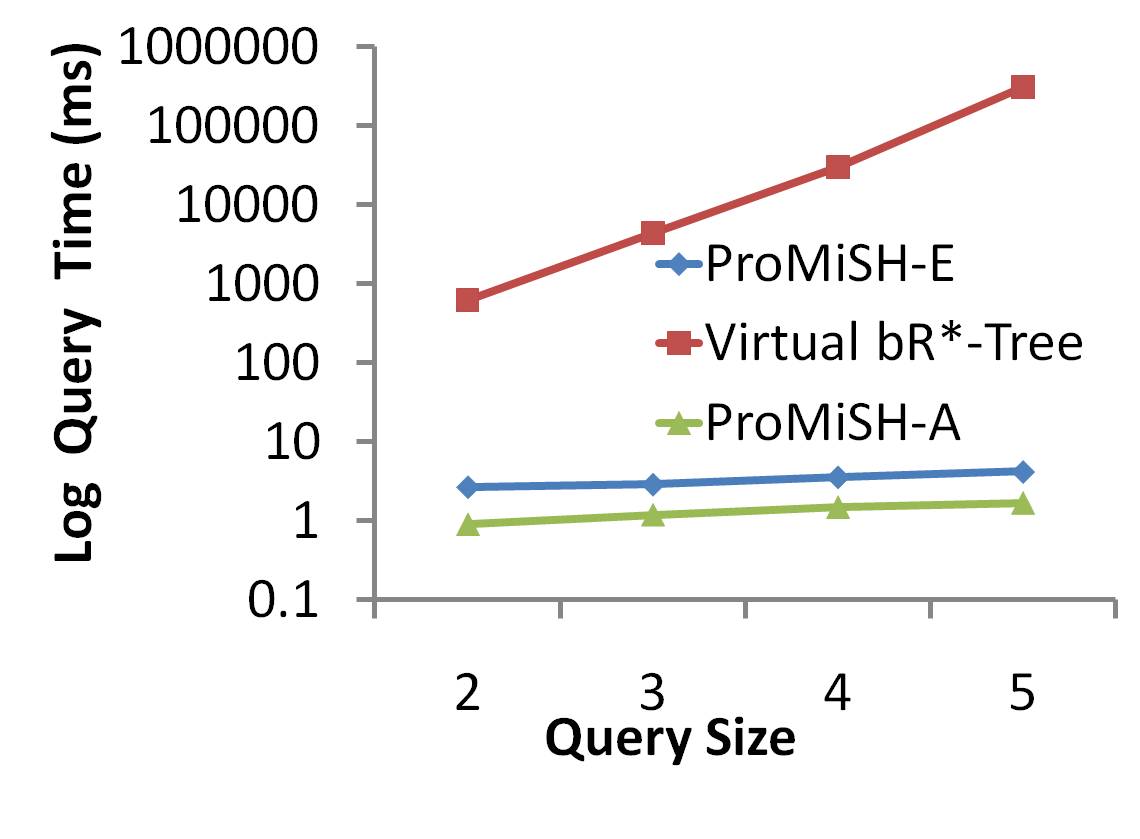}
\vspace{-4mm}
\caption{Query time comparison of algorithms for retrieving top-$1$ results for queries of varying sizes $q$ on a $10$-dimensional synthetic dataset having $100$,$000$ points. Values of $t$=$1$ and $U$=$1$,$000$ were used for the dataset.}
\vspace{-2mm}
\label{fig:syn_comp_var_q}
\end{figure}

\begin{figure}[t!]
\centering
\includegraphics[width=0.7\columnwidth]{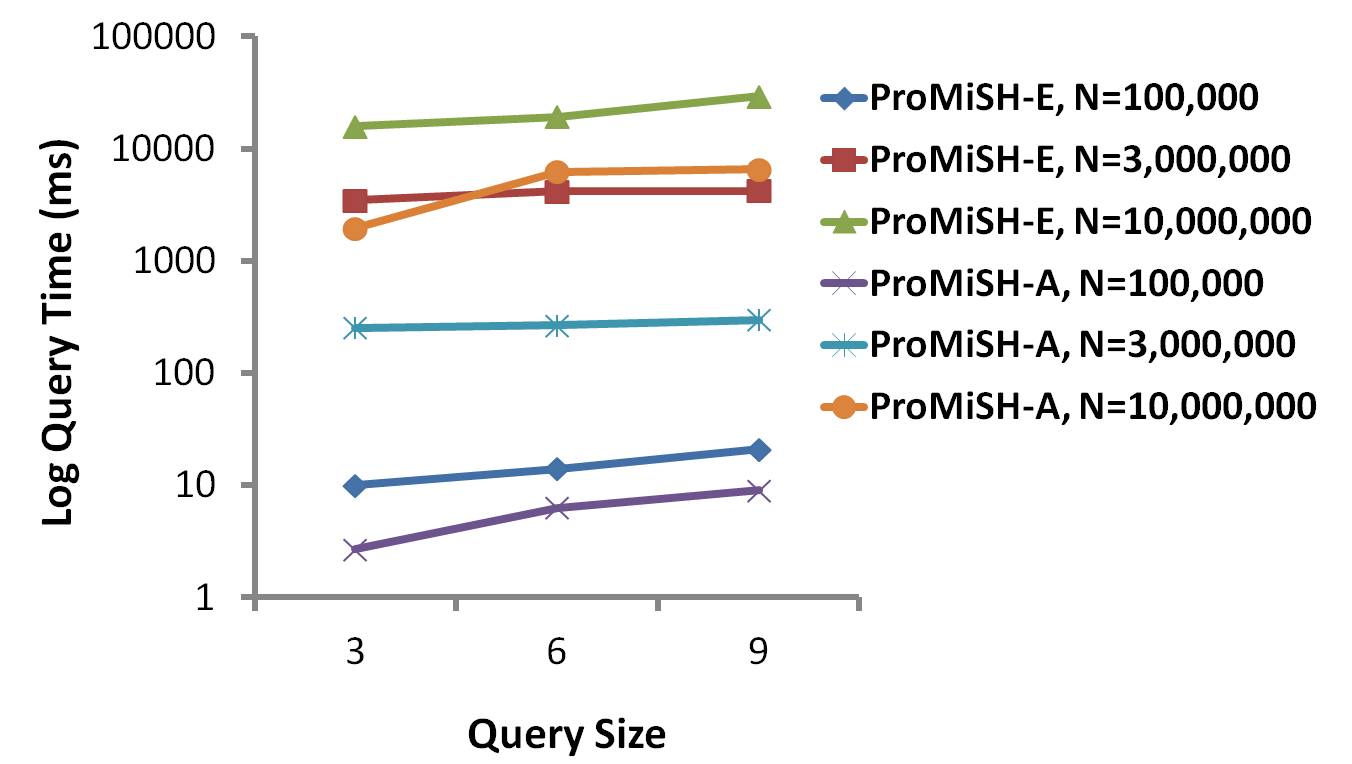}
\vspace{-4mm}
\caption{Query time analysis of ProMiSH algorithms for retrieving top-$1$ results for queries of varying sizes $q$ on $25$-dimensional synthetic datasets of varying sizes $N$. Values of $t$=$1$ and $U$=$200$ were used for each dataset.}
\vspace{-5mm}
\label{fig:syn_ea_var_q}
\end{figure}

\subsection{Quality Test}
\label{lb:quality}
We validated the result quality of ProMiSH-E, ProMiSH-A and Virtual bR*-Tree by their average approximation ratios (AAR). ProM\-iSH-E and Virtual bR*-Tree perform an exact search. Therefore, they always retrieve the true top-$k$ results, and have AAR of $1$. We used the results returned by them as the ground truth. Figure~\ref{fig:approx_ratio} shows AAR computed over top-$5$ results retrieved by ProMiSH-A for varying query sizes on two $32$-dimensional real datasets. We observe from figure~\ref{fig:approx_ratio} that AAR of ProMiSH-A is always less than $1.5$. This low AAR allows ProMiSH-A to return practically useful results with a very efficient time and space complexity.

\subsection{Efficiency on Synthetic Datasets}
We performed experiments on multiple synthetic datasets to verify the efficiency and the scalability of ProMiSH. We first discuss the comparison of query times of Virtual bR*-Tree, ProMiSH-A, and ProMiSH-E for varying dataset dimensions $d$, dataset sizes $N$, and query sizes $q$. We found that ProMiSH performs at least four orders of magnitude better than Virtual bR*-Tree. We also show results of the scalability tests of ProMiSH for varying values of $N$, $d$, $q$, and the result size $k$. Our scalability results reveal a linear performance of ProMiSH with $N$, $d$, $q$, and $k$. All the query times are measured in milliseconds (ms) and  shown in log scale in all the figures.

The query times of ProMiSH-E, ProMiSH-A, and Virtual bR*-Tree for retrieving top-$1$ results for queries of size $5$ on datasets of varying dimensions $d$ are shown in figure~\ref{fig:syn_comp_var_d}. We used a dataset of $100$,$000$ points where each point was tagged with $t$=$1$ keyword using a dictionary of size  $U$=$1$,$000$. For the dataset of dimension $25$, ProMiSH-A completed in  $1.8$ ms and ProMiSH-E took only $4.2$ ms. Conversely, results for Virtual  bR*-Tree could not be obtained since it ran for more than $5$ hours. We observed that ProMiSH not only significantly outperforms Virtual bR*-Tree on datasets of all dimensions but the difference in performance also grows to more than five orders with an increase in the dataset dimension.

\begin{figure}[t!]
\centering
\includegraphics[width=0.7\columnwidth]{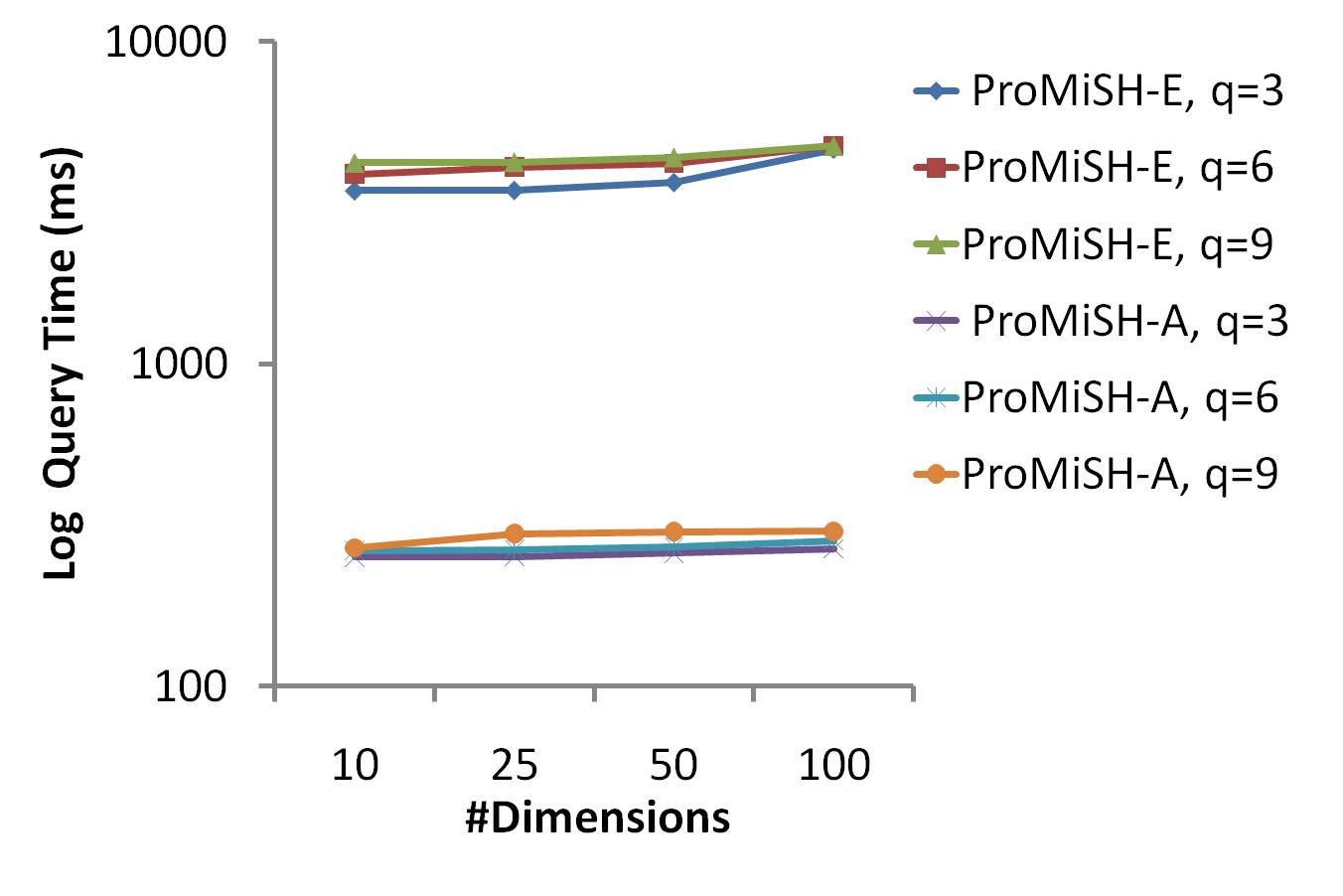}
\vspace{-5mm}
\caption{Query time analysis of ProMiSH for retrieving top-$1$ results for queries of varying sizes $q$ on large synthetic datasets of varying dimensions $d$. Values of $N$=$3$ million, $t$=$1$, and $U$=$200$ were used for each dataset.}
\vspace{-2mm}
\label{fig:syn_ea_var_d}
\end{figure}

\begin{figure}[t!]
\centering
\includegraphics[width=0.55\columnwidth]{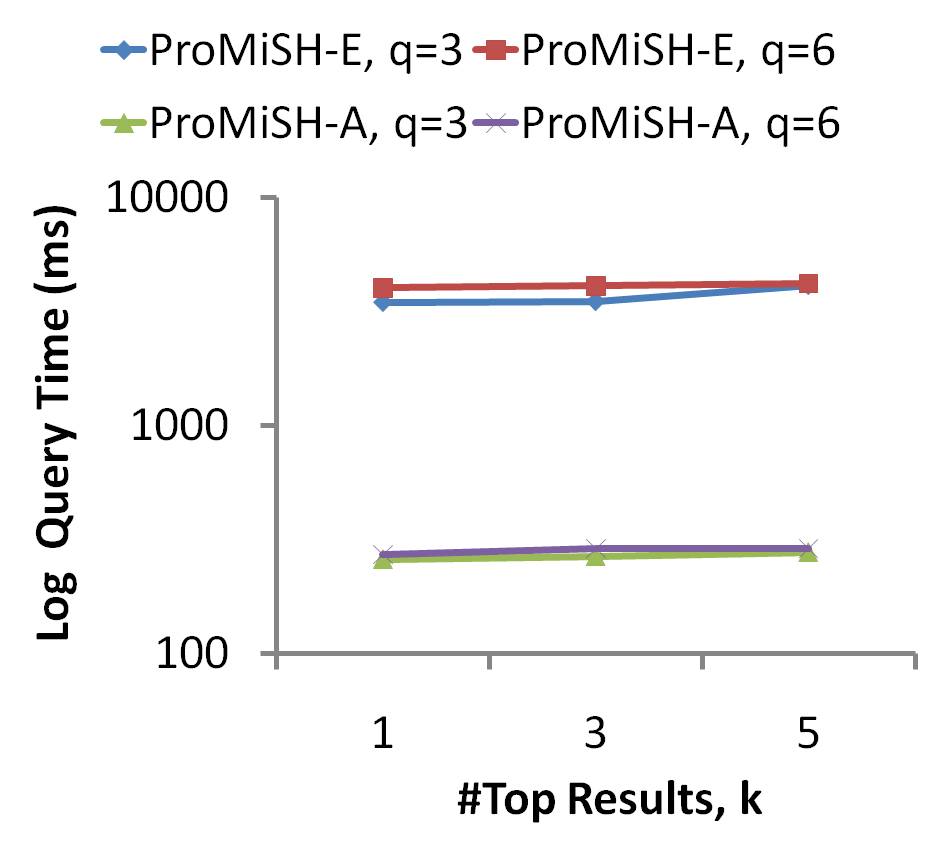}
\vspace{-4mm}
\caption{Query time analysis of ProMiSH algorithms for retrieving top-$k$ results for queries of sizes $3$ and $6$ on a $50$-dimensional synthetic dataset of size $N$=$3$ million. Values of $t$=$1$ and $U$=$200$ were used for the dataset.}
\vspace{-3mm}
\label{fig:syn_ea_var_k}
\end{figure}

\begin{figure}
\centering
\includegraphics[width=0.5\columnwidth]{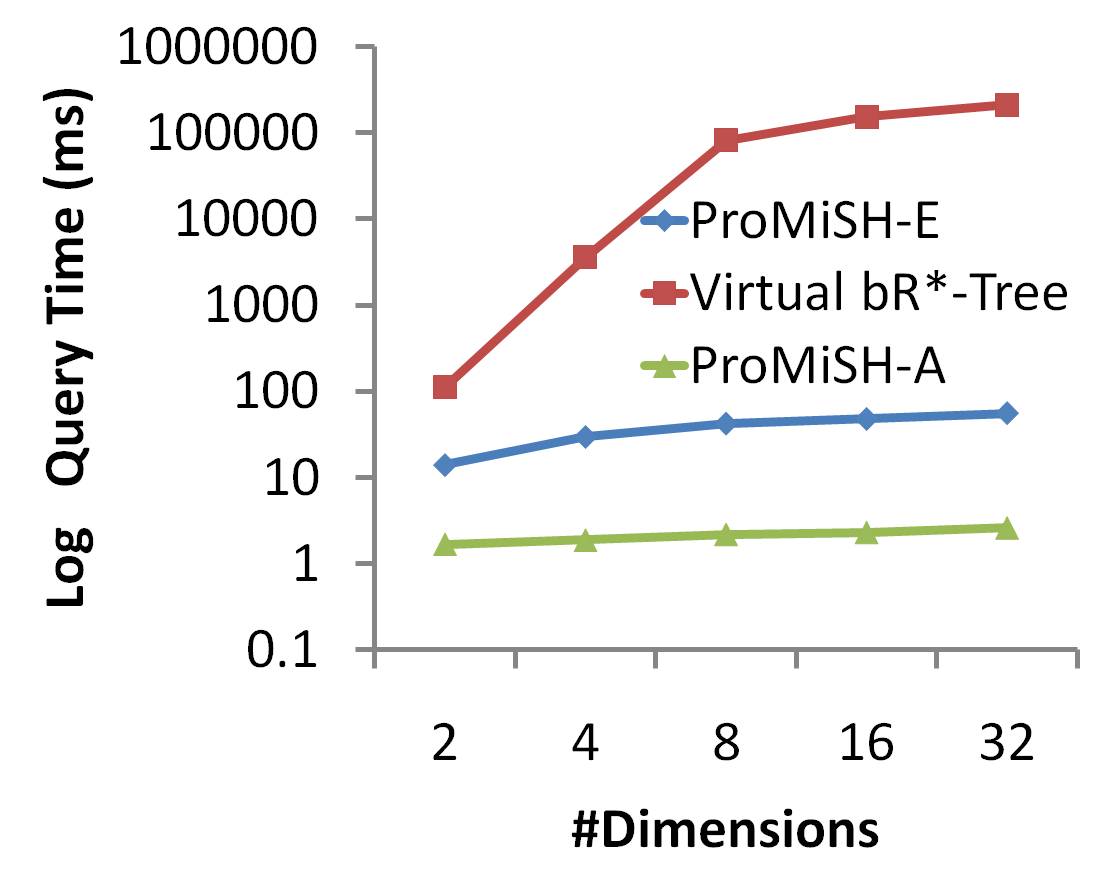}
\vspace{-4mm}
\caption{Query time comparison of algorithms for retrieving top-$1$ results for queries of size $q$=$4$ on real datasets of varying dimensions $d$ and size $N$=$50$,$000$.}
\vspace{-5mm}
\label{fig:rd_comp_var_d}
\end{figure}

\begin{figure}
\centering
\includegraphics[width=0.5\columnwidth]{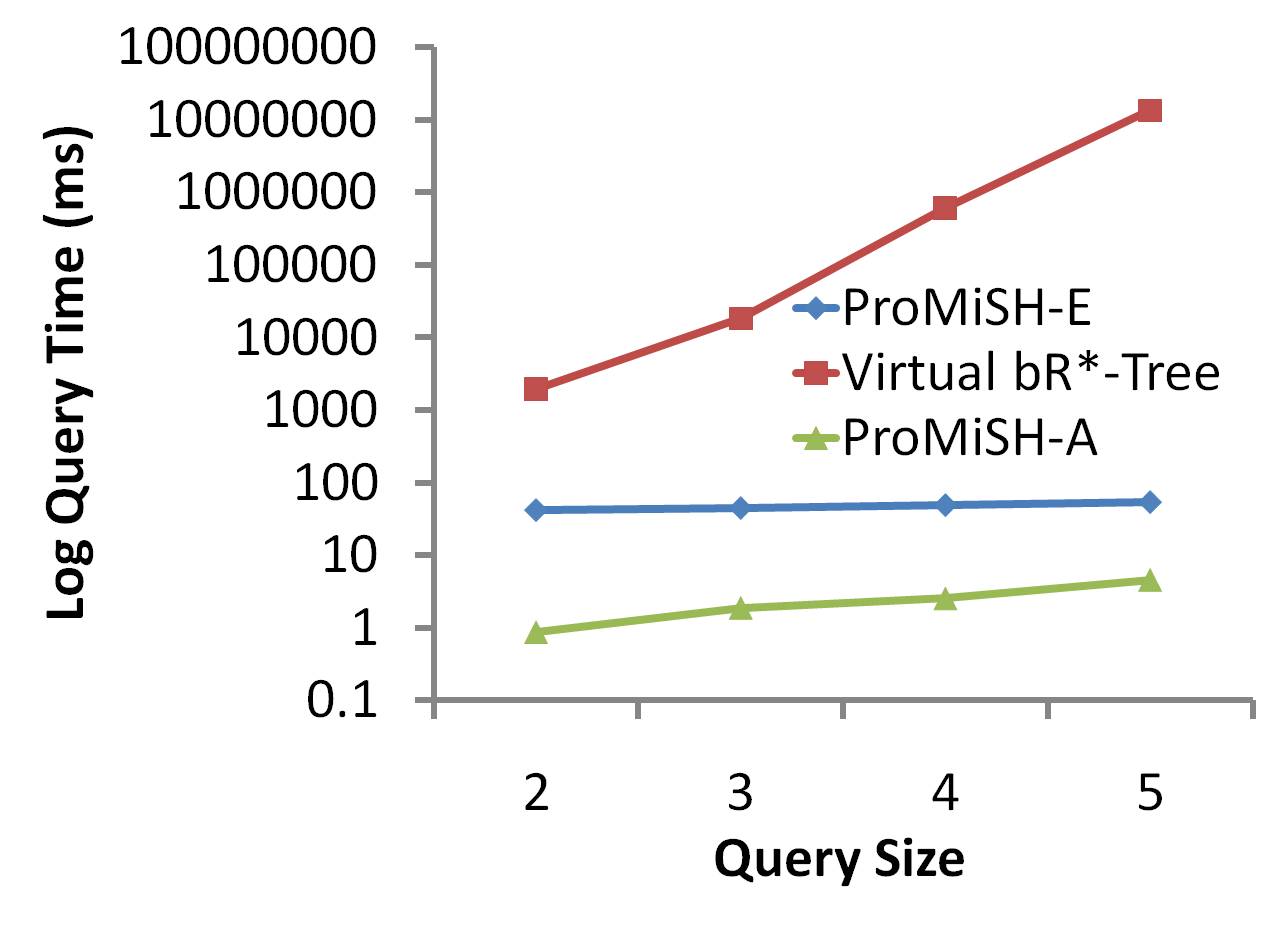}
\vspace{-4mm}
\caption{Query time comparison of algorithms for retrieving top-$1$ results for queries of varying sizes $q$ on a $16$-dimensional real dataset of size $N$=$70$,$000$.} 
\vspace{-3mm}
\label{fig:rd_comp_var_q}
\end{figure}

We show the query times of the algorithms on $25$-dime\-nsional datasets of varying sizes $N$ for queries of size $5$ in figure~\ref{fig:syn_comp_var_N}. Each dataset used a dictionary of size $U$=$1$,$000$ and $t$=$1$ keyword per point. Virtual bR*-Tree failed to finish for the dataset of size $N$=$100$,$000$ even after $5$ hours of execution. We report the query times of the algorithms for queries of varying sizes $q$ on a $10$-dimensional dataset of size $N$=$100$,$000$ in figure~\ref{fig:syn_comp_var_q}. Each data point was tagged with $t$=$1$ keyword using a dictionary of size $U$=$1$,$000$. For a query of size $5$, ProMiSH-A had a query time of $1.7$ ms, ProMiSH-E had a query time of $4.2$ ms, and Virtual bR*-Tree had a query time of $305$ seconds. We again observed that ProMiSH outperforms Virtual bR*-Tree by more than five orders of magnitude with an increase in the dataset size and the query size.

All the above results show that the query time of ProMiSH increases linearly with the dataset size $N$, the dataset dimension $d$, and the query size $q$. In contrast, Virtual bR*-Tree fails to scale with $q$, $d$, and $N$. These results confirm that the pruning criteria of Virtual bR*-Tree, as discussed in section~\ref{subsec:survey}, becomes ineffective with an increase in the dimension of the dataset. This leads to an exponential generation of potential candidates and large query times.

Next, we present scalability results of ProMiSH-E and ProMiSH-A on large synthetic datasets of varying dimensions for large query sizes and varying result sizes. Each dataset used a dictionary of size $U$=$200$. A point in each dataset was tagged with $t$=$1$ keyword. Figure~\ref{fig:syn_ea_var_q} shows the query times for queries of varying sizes $q$ on $25$-dimensional datasets of varying sizes $N$. ProMiSH-E had a query time of $29$ seconds and ProMiSH-A had a query time of $6$ seconds for queries of size $9$ on a dataset of $10$ million points. We observed that ProMiSH-A is an order of magnitude faster than ProMiSH-E for queries of all sizes. We see from figure~\ref{fig:syn_ea_var_q} that ProMiSH scales linearly with the query size and the dataset size.

\begin{figure}
\centering
\includegraphics[width=0.5\columnwidth]{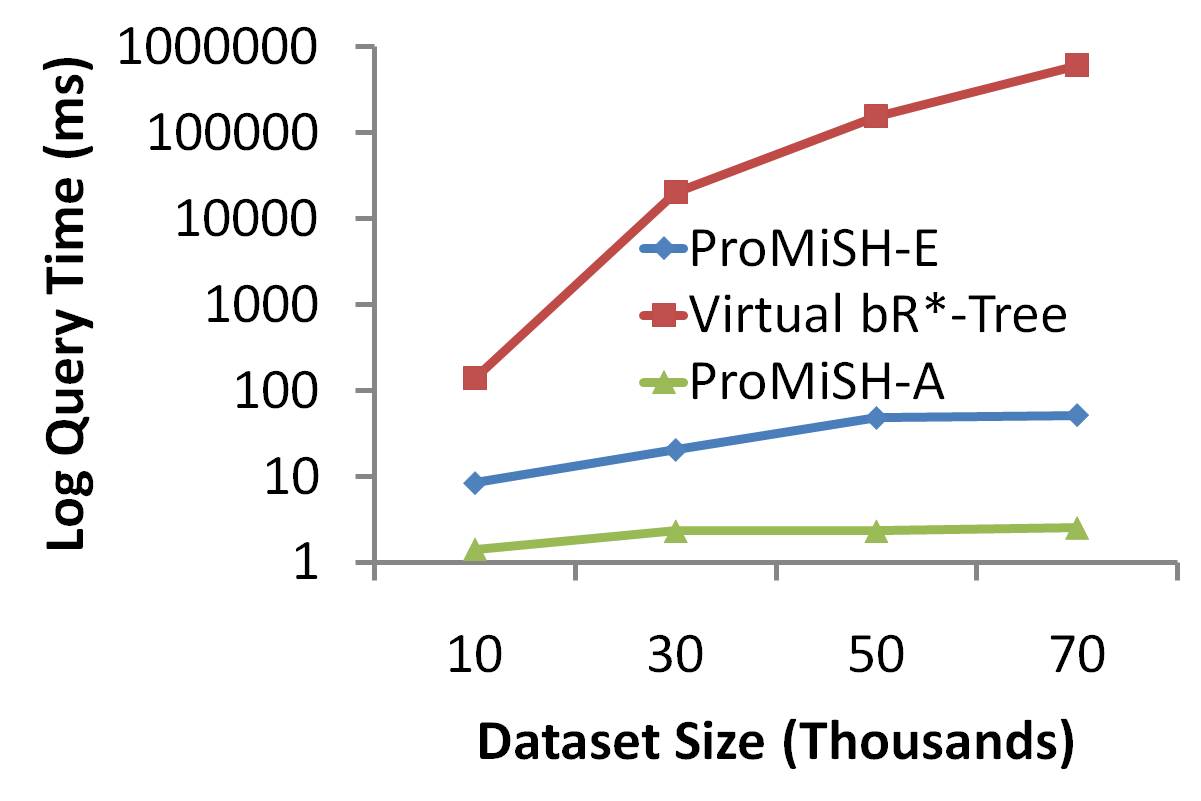}
\vspace{-4mm}
\caption{Query time comparison of algorithms for retrieving top-$1$ results for queries of size $q$=$4$ on $16$-dimensional real datasets of varying sizes $N$.}
\vspace{-3mm}
\label{fig:rd_comp_var_N}
\end{figure}

\begin{figure}
\centering
\includegraphics[width=0.52\columnwidth]{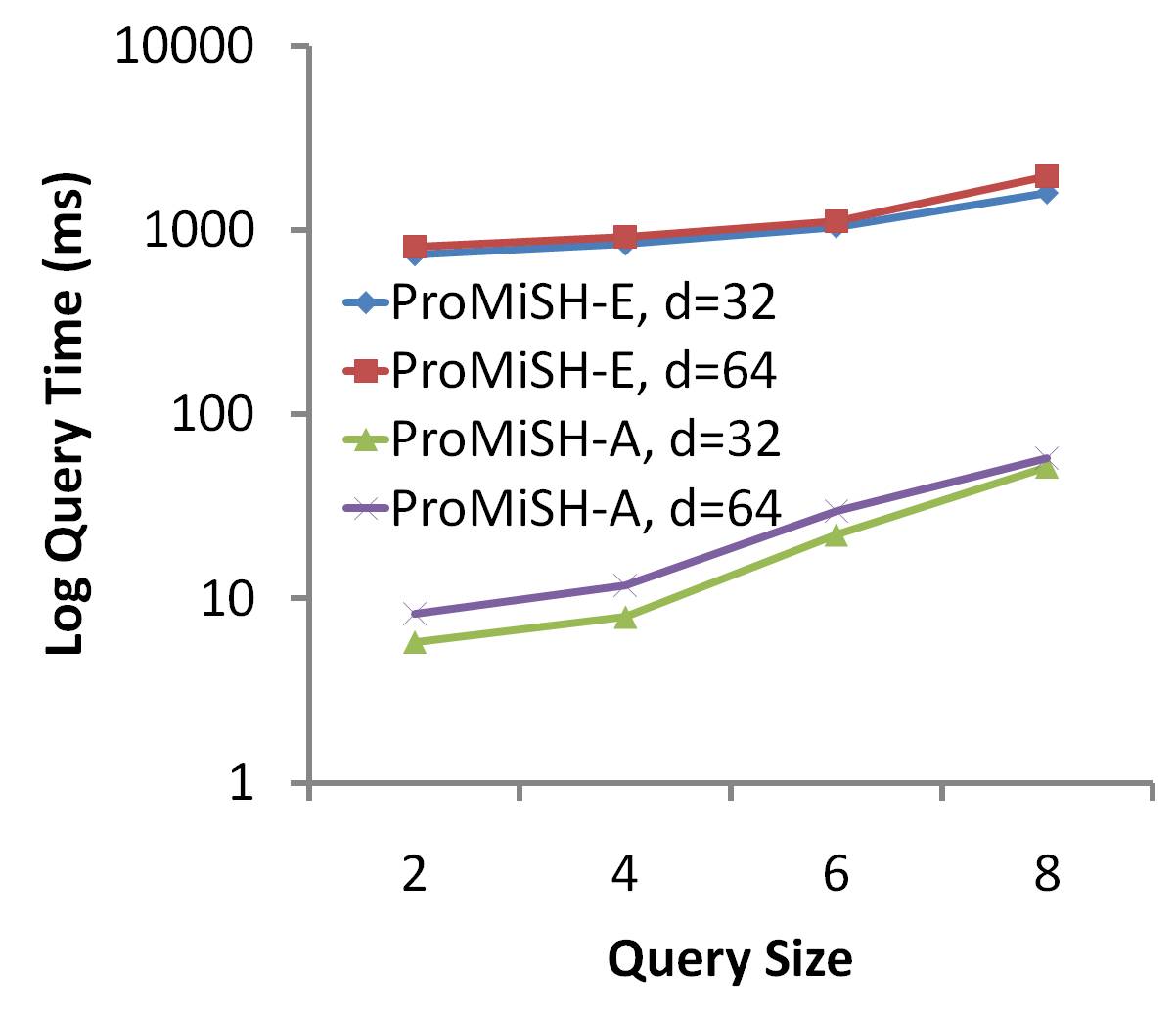}
\vspace{-4mm}
\caption{Query time analysis of ProMiSH algorithms for retrieving top-$1$ results for queries of varying sizes $q$ on real datasets of varying dimensions and size $N$=$1$ million.}
\vspace{-5mm}
\label{fig:rd_ea_var_q}
\end{figure}

Figure~\ref{fig:syn_ea_var_d} shows the query times of ProMiSH for queries of varying sizes on $3$ million size datasets of varying dimensions. ProMiSH-E had a query time of $4.7$ seconds and ProMiSH-A had a query time of $0.3$ seconds for queries of size $q$=$9$ on a $100$-dimensional dataset. ProMiSH-A is an order of magnitude faster than ProMiSH-E on datasets of all dimensions. We observed that both algorithms scale linearly with dimension $d$ of the dataset. Figure~\ref{fig:syn_ea_var_k} shows the query times for retrieving the top-$k$ results for queries of varying  sizes $q$ on a $50$-dimensional dataset. It reveals a linear performance of   both algorithms for increasing $k$. ProMiSH-A is an order of magnitude better than ProMiSH-E for any result size $k$. All these tests show that the query time of ProMiSH scales linearly with the dataset size, the dataset dimension, the query size, and the result size.

\subsection{Efficiency on Real Datasets}
We evaluated the efficiency and the scalability of ProMiSH on multiple real datasets. We first discuss query time comparisons of alternative algorithms for varying dataset dimensions $d$,  query sizes $q$, and dataset sizes $N$. We also discuss scalability tests of ProMiSH-E and ProMiSH-A for varying values of $q$, $d$, and the result size $k$. All the query times are measured in milliseconds (ms) and  shown in log scale in all the figures.

We show the query times of the algorithms on real datasets of varying dimensions $d$ in figure~\ref{fig:rd_comp_var_d}. We used datasets of size $N$=$50$,$000$ and queries of size $q$=$4$. ProMiSH-A had a query time of $3$ ms, ProMiSH-E had a query time of $55$ ms, and Virtual bR*-Tree had a query time of $210$ seconds for $32$-dimensional dataset. Comparison of query times for queries of varying sizes $q$ on a $16$-dimensional real dataset of size $N$=$70$,$000$ is shown in figure~\ref{fig:rd_comp_var_q}. ProMiSH-A had a query time of $5$ ms, ProMiSH-E had a query time of $54$ ms, and Virtual bR*-Tree had a query time of $13$,$352$ seconds for queries of size $q$=$5$. Comparison of query times on $16$-dimensional real datasets of varying sizes $N$ for queries of size $q$=$4$ is shown in figure~\ref{fig:rd_comp_var_N}. ProMiSH-A had a query time of $3$ ms, ProMiSH-E had a query time of $49$ ms, and Virtual bR*-Tree had a query time of $608$ seconds for a dataset of size $N$=$70$,$000$.

\begin{figure}
\centering
\includegraphics[width=0.52\columnwidth]{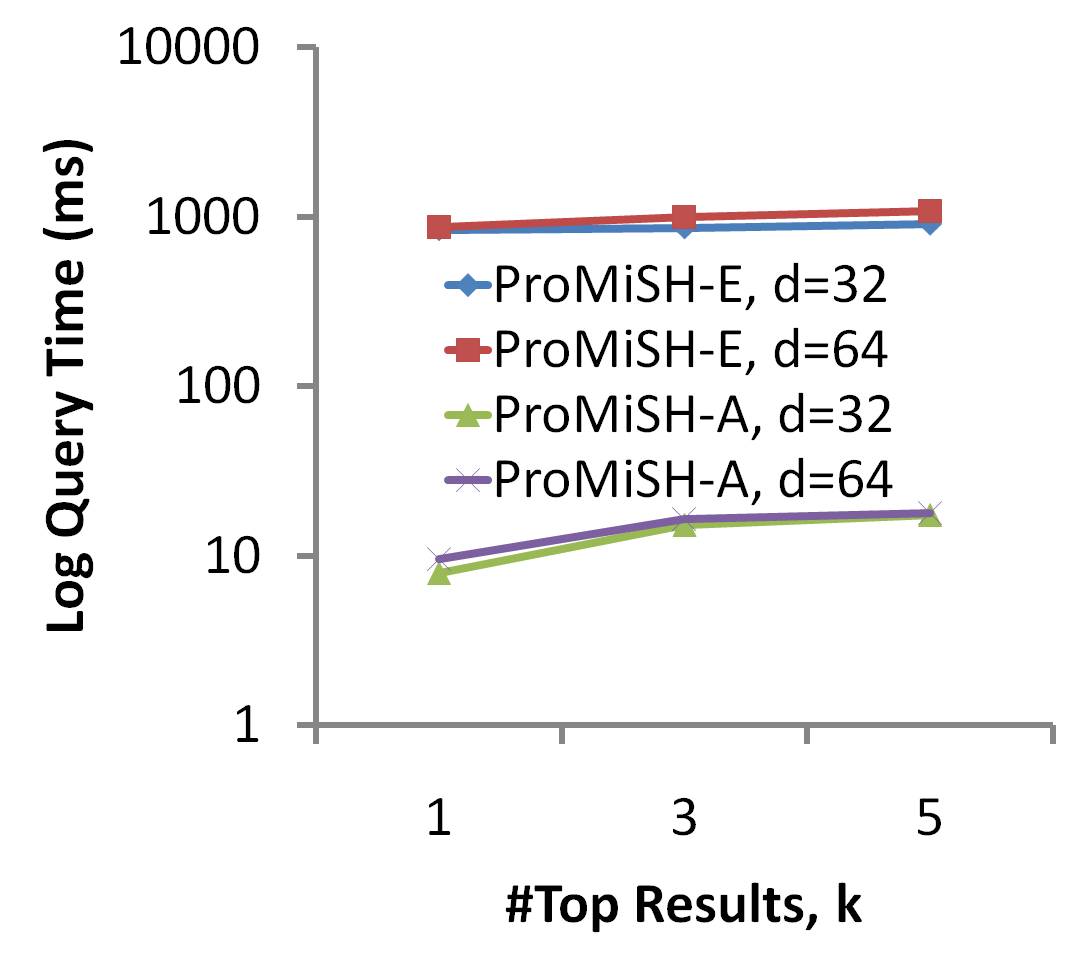}
\vspace{-4mm}
\caption{Query time analysis of ProMiSH algorithms for retrieving top-$k$ results for queries of size $q$=$4$ on real datasets of varying dimensions and size $N$=$1$ million.}
\vspace{-5mm}
\label{fig:rd_ea_var_k_t1}
\end{figure}

\begin{table*}[t]
\begin{center}
    \scriptsize
    \begin{tabular}{|c|c|c|c|c|c|c|c|c|c|c|c|c|}
     \hline
     & \multicolumn{4}{|c|}{ProMiSH-E} &\multicolumn{4}{|c|}{ProMiSH-A} & \multicolumn{4}{|c|}{Virtual bR*-Tree}\\
    \hline
     & \multicolumn{2}{|c|}{$N$=$10$ million} & \multicolumn{2}{|c|}{$N$=$100$ million}& \multicolumn{2}{|c|}{$N$=$10$ million}& \multicolumn{2}{|c|}{$N$=$100$ million} & \multicolumn{2}{|c|}{$N$=$10$ million}& \multicolumn{2}{|c|}{$N$=$100$ million}\\
    \hline
    $d$ &$U$=$100$& $U$=$1$,$000$& $U$=$100$& $U$=$1$,$000$&$U$=$100$& $U$=$1$,$000$&$U$=$100$& $U$=$1$,$000$& $U$=$100$& $U$=$1$,$000$& $U$=$100$& $U$=$1$,$000$\\
    \hline
    $8$ &2.8&3.0&2.8&2.8&0.7&0.9&0.7&0.7& 0.9&0.9 &0.9 &0.9 \\
    \hline
    $16$ &1.4&1.6&1.5&1.5&0.4&0.5&0.4&0.4& 0.5& 0.5& 0.6& 0.6\\
    \hline
    $32$ &0.7&0.8&0.8&0.8&0.2&0.2&0.2&0.2& 0.3&0.3 & 0.4& 0.4\\
    \hline
    $64$ &0.4&0.4&0.4&0.4&0.09&0.1&0.09&0.09&0.2 & 0.2&0.3&0.3\\
    \hline
    $128$&0.2&0.2&0.2&0.2&0.05&0.06&0.05&0.05&0.2 & 0.2&0.2 &0.2\\
    \hline
    \end{tabular}
    \end{center}
    \vspace*{0mm}
    \caption{Ratio of the index space to the dataset space for ProMiSH-E, ProMiSH-A and Virtual bR*-Tree for varying $N$, $d$, and $U$.}
    \vspace*{-5mm}
    \label{tab:spaceRatioE}
\end{table*}

The above results show that ProMiSH significantly outperforms state-of-the-art Virtual bR*-Tree on real datasets of all dimensions and sizes and on queries of all sizes. ProMiSH-E is five orders of magnitude faster than Virtual bR*-Tree for queries of size $q$=$5$ on a $16$-dimensional real dataset of size $70$,$000$.  ProMiSH-E is also at least four orders of magnitude faster than Virtual bR*-Tree for a queries of size  $q$=$4$ on a $32$-dimensional real dataset of size $50$,$000$. ProMiSH-A always has an order of magnitude better performance than ProMiSH-E. Similar to the observations on synthetic datasets, we find that the difference in query time of ProMiSH and Virtual bR*-Tree grows to multiple orders of magnitude with an increase in the dataset size $N$, the dataset dimension $d$, and the query size $q$. In addition, the query time performance of ProMiSH-A and ProMiSH-E is linear with the dataset size, the dataset dimension, and the query size, unlike Virtual bR*-Tree whose performance deteriorates sharply. This again confirms that the pruning criteria of Virtual bR*-Tree is ineffective for high-dimensional datasets.

We performed stress tests of ProMiSH on real datasets having $1$ million points of dimensions $32$ and $64$. Figure~\ref{fig:rd_ea_var_q} shows the query times of ProMiSH-A and ProMiSH-E for varying query sizes. Pro\-MiSH-A had a query time of $58$ ms and ProMiSH-E had a query time of $1$,$592$ ms for queries of size $q$=$8$ on $64$-dimensional real datasets. Figure~\ref{fig:rd_ea_var_k_t1} shows the query times of ProMiSH for retrieving top-$k$ results for queries of size $q$=$4$. ProMiSH-A had a query time of $18$ ms and ProMiSH-E had a query time of $1$,$084$ ms for top-$5$ results on $64$-dimensional datasets. Figures~\ref{fig:rd_ea_var_q} and \ref{fig:rd_ea_var_k_t1} verify that the query time of ProMiSH increases linearly  with $d$, $q$, and $k$.

Our evaluations on real datasets of high dimensions establish that ProMiSH  scales linearly with the dataset size, the dataset dimension, the query size, and the result size. ProMiSH also yields practical query times on large datasets of high dimensions, and is very useful for answering real time queries.

\subsection{Space Efficiency}
\label{sec:promishESC}
We evaluated the space efficiency of ProMiSH by computing the memory footprint of its index. For ProMiSH-E and ProMiSH-A, we used the space cost formulation from  section~\ref{sec:costAnal}. Here we first describe the space cost of Virtual-bR* Tree in terms of the dataset and the index parameters. Then, we give ratios of the index space to the dataset space for all the three algorithms for varying dataset parameters. Let the space cost of a point's identifier, a dimension of a point, and a keyword be $E$ bytes individually. Let $\mathcal{D}$ be a dataset having $N$ $d$-dimensional points each of which is tagged with $t$ keywords. Let $U$ be the number of unique keywords in $\mathcal{D}$. The dataset has a space cost of $S(\mathcal{D})$=(($d$+$t$)$\times$ $N$ $\times$ $E$) bytes.


Index of Virtual bR*-Tree comprises of a R*-Tree, an inverted index, and a  bR*-Tree. Let the number of children per node in R*-Tree be $x$. Let the total number of nodes in R*-Tree be $N_R$. The space cost of R*-Tree is $((2 \times d + x) \times E \times N_R)$ bytes. The inverted index stores a point's identifier and  its path from the root node in R*-Tree. Therefore, the space cost of the inverted index is $((log_x N + 1) \times t \times E \times N)$ bytes. For a query of size $q$, the space cost of bR*-Tree  is $((2 \times d \times E + 2 \times d \times E \times q + x \times E+ U/8) \times N_R)$ bytes.

We investigated ratios of the index space to the dataset space for all the three algorithms using their space cost formulations. We used following values of the parameters: $E$=$4$ bytes, $m$=$2$, $M$=$10$,$000$, $L$=$5$, $x$=$100$, $q$=$5$, and $t$=$1$. We show the ratios for varying values of $d$, $N$, and $U$ in table~\ref{tab:spaceRatioE}. For datasets of low dimensions, e.g., $d$=$8$, we observe from table~\ref{tab:spaceRatioE} that ProMiSH-E has the highest ratios, whereas ProMiSH-A has the lowest ratios. For datasets of high dimensions, e.g., $d$=$128$, we observe from table~\ref{tab:spaceRatioE} that ProMiSH-E and Virtual bR*-Tree have comparable ratios, whereas ProMiSH-A again has the lowest ratios.

We see that the index space of ProMiSH is independent of the dimension, whereas the dataset space grows linearly with it. Therefore, the space ratio of ProMiSH decreases with dimension. The index space of Virtual bR*-Tree also grows with dimension. Therefore, ProMiSH has a lower space ratio than Virtual bR*-Tree for high dimensions.




\section{Extending ProMiSH to Disk}
\label{disk}
Here we discuss extension of ProMiSH to disk. As seen from Algorithm~\ref{alg:ProMiSHE}, ProMiSH-E sequentially reads only required buckets from $\mathcal{I}_{kp}$ to find points containing at least one query keyword.  Therefore, $\mathcal{I}_{kp}$ is stored on disk using a directory-file structure. A directory is created for $\mathcal{I}_{kp}$. Each bucket of $\mathcal{I}_{kp}$ is stored in a separate file named after its key in the directory. We also see from Algorithm~\ref{alg:ProMiSHE} that ProMiSH sequentially probes $\mathcal{HI}$ data structures starting at the smallest scale to generate the candidate point ids for the subset search. Further, it reads only required buckets from the hashtable and the inverted index of a $\mathcal{HI}$ structure. Therefore, all the hashtables and the inverted indices of $\mathcal{HI}$ are again stored using a similar directory-file structure as $\mathcal{I}_{kp}$. All the points in the dataset are indexed into a B+-Tree~\cite{BTree} using their ids and stored on the disk. Subset search retrieves the points from the disk using B+-Tree for exploring the final set of results. 

\section{Conclusions and Future Work}
\label{conclusions}
In this paper, we proposed solutions for the problem of top-$k$ nearest keyword set search in multi-dimensional datasets. We developed  an exact (ProMiSH-E) and an approximate (ProMiSH-A) method. We designed a novel index based on random projections and hashing. Index is used to find subset of points containing the true results. We also proposed an efficient solution to query results from a subset of data points. Our empirical results show that ProMiSH is faster than state-of-the-art tree-based technique, having performance improvements of multiple orders of magnitude. These performance gains are further emphasized as dataset size and dimension increase, as well as for large query sizes. ProMiSH-A has the fastest query time. We empirically observed a linear scalability of ProMiSH with the dataset size, the dataset dimension, the query size, and the result size. We also observed that ProMiSH yield practical query times on large datasets of high dimensions for queries of large sizes.

In the future, we plan to explore other scoring schemes for ranking the result sets.  In one scheme, we may assign weights to the keywords of a point by using techniques like tf-idf. Then, each group of points can be scored  based both on the distance between the points and weights of the keywords. Further, the criteria of a result containing all the keywords can be relaxed to generate results having only a subset of the query keywords.  

\section{Acknowledgments}
This research was supported partially by the National Science Foundation under grant IIS-1219254.

\small

\bibliographystyle{IEEEtran}
\bibliography{keywordSetSearch}
\end{document}